\pdfoutput=1
\documentclass[11pt]{article}
\usepackage[ruled, vlined, linesnumbered, noresetcount]{algorithm2e}
\usepackage{times}
\usepackage{epsfig}
\usepackage{amsthm}
\usepackage{amsmath}
\usepackage{amsfonts}
\usepackage{amssymb}
\usepackage{enumerate}
\usepackage{graphics}
\usepackage{color}
\usepackage{mathtools}
\usepackage{url}

\SetKwFor{ForAll}{forall}{do}{end forall}

\usepackage{graphicx}

\date{}

\SetKwFor{ForAll}{forall}{do}{end forall}

\newtheorem{theorem}{Theorem}
\newtheorem{example}{Example}
\newtheorem{lemma}{Lemma}

\newtheorem{corollary}{Corollary}

\newtheorem{fact}{Fact}
\newtheorem{remk}{Remark}

\def\iin{{in}}
\def\oout{{out}}
\def\d{{\delta}}

\def\la{{\lambda}}

\newcommand{\remove}[1]{}
\newcommand{\N}{{\mathbb{N}}}
\newcommand{\Active}{{{\sf Active}}}

\newcommand\myatop[2]{\genfrac{}{}{0pt}{}{#1}{#2}}

\begin{document}

\title{On Finding Small Sets that Influence \\Large Networks\thanks{A preliminary version of this paper was presented at the  1st International Workshop on Dynamics in Networks (DyNo 2015)
in conjunction with the 2016 IEEE/ACM International Conference ASONAM, Paris, France, August 25-28, 2015. \cite{CordascoGR15}}
}

\author{Gennaro Cordasco\\Dept.~of Psychology\\ Second University of Naples, Italy 
\and
Luisa Gargano\\ Dept.~ of Computer Science \\ University of Salerno, Italy
\and 
Adele Anna Rescigno \\ Dept.~ of Computer Science \\ University of Salerno, Italy}

\setcounter{footnote}{0}
\maketitle 
\begin{abstract}
We consider the problem of selecting a minimum size
subset of nodes in a network, that  allows to activate
all the nodes of the network. We present a fast and simple
algorithm that, in real-life networks, produces solutions that
outperform the ones obtained by using the best  algorithms in the 
literature.
We  also investigate  the theoretical performances
of our algorithm and give   proofs of optimality 
for some classes of graphs. 
From an experimental perspective,  experiments 
 also show that the performance of the algorithms
correlates with the modularity of the analyzed network. 
Moreover, the more  the influence among communities  is hard to propagate, 
the  less the performances of the  algorithms differ. 
On the other hand, when  the network allows some propagation of
influence between different communities, the gap between the solutions returned by the
proposed  algorithm and by the previous algorithms  in the literature increases.
%\keywords{ Social Networks\and Spread of Influence \and Viral Marketing \and MinimumTarget Set}
\end{abstract}

\section{Introduction }
The study of networked phenomena has experienced a particular surge of interest over the past decade, thanks to the large diffusion of social networks   which led to increasing availability of huge amounts of data in terms of static network topology as well as the dynamic of interactions among users \cite{WF}.   
\\
A large part of such studies deals with  the analysis of influence spreading in social networks. 
Social influence is the process by which individuals, interacting with other people, change or adapt their attitude, belief or behavior in order to fit in with a group \cite{CF}.
Commercial companies, as well as  politician,   have soon recognized that they can benefit from a social influence process which advertises their product (or belief) from one person to another \cite{Bo+,LKLG,T,RW} . This advertising process is well known as {viral marketing} \cite{LAM}.
A key research question in the area of viral marketing is how to efficiently identify a set of users which are able to widely disseminate a certain information within the network.
This matter suggests several  optimization  problems. Some of them  were first articulated in the seminal papers
\cite{DR-01,KKT-03,KKT-05}, under various influence propagation models.
A description of the area can be found in the recent monograph \cite{BOOK2013}.

In this paper we consider the Minimum Target Set problem which,  roughly speaking,  asks for selecting a minimum size
subset of nodes in a network that once active are able   to activate
all the nodes of the network under the Linear Threshold (LT)   influence propagation model.  
According to the  LT model,
%     has recently received much attention his model,  
 a user $v$ becomes active when the sum of influences of its neighbors in the networks reaches a certain threshold $t(v)$ \cite{KKT-15}. The formal definition of the Minimum Target Set problem is  given  in Section \ref{secMTS}.
%
%Several papers analyzed the IC  model \cite{W+,Xu2016110}. In particular, Kundu \emph{et al.} \cite{Kundu2015107} studied the Maximum Influencing Set problem under the  IC  model. They introduced a deprecation based approach where
%, instead of targeting the more influencing nodes, 
%the algorithm iteratively deprecates (i.e., removes from the graph) the less influencing nodes.  The target set is determined, in this case, by the nodes that remain in the graph.
%
\\
Chen \cite{Chen-09} studied the Minimum Target Set problem under the LT model    and proved  a strong inapproximability result that makes unlikely the existence
of an  algorithm with  approximation factor better than  $O(2^{\log^{1-\epsilon }n})$, where $n$ is the number of nodes in the network.
Chen's result stimulated a series of papers  
 that isolated interesting cases 
in which the problem (and variants thereof) becomes tractable
\cite{ABW-10,BCNS,BHLM-11,Centeno12,Chun,Chun2,Chopin-12,Cic+,Cic14,C-OFKR,CGRV2015,CGRV16,CordascoGRV16,Ga+,NNUW,Re,Za}.
 \\
In the case of \emph{general  networks}, some 
efficient heuristics for the Minimum Target Set  problem  have been proposed in the literature 
\cite{CGM+,DZNT,SEP}. 
 In particular,  Shakarian {\em et al.} \cite{SEP}  introduced a deprecation based approach where 
the algorithm iteratively deprecates (i.e., removes from the network) the less influencing nodes.  The output  set is determined, in this case, by the nodes that remain in the network.
Subsequently, the authors of  \cite{CGM+} proposed   a novel deprecation-like approach  which, from the theoretically point of view, always produces optimal solution
(i.e, a minimum size subset of nodes that influence the whole network)  for trees, cycles and cliques and on real life networks produces solutions that outperform the ones obtained using
 previous algorithms   \cite{DZNT,SEP}.

\subsection{Our Results}
In this paper we present an evolution of the heuristics  in \cite{CGM+}. 
It   is an extension from  undirected   to directed networks that allows  the additional feature of taking  into account the influence that a deprecated node may apply on his outgoing neighbors. 
This extension allows to strongly improve the quality of the obtained solution. 
Indeed, in the  previous  deprecation-based algorithms for the Minimum Target Set problem, once a node was determined  as irrelevant, it was immediately pruned from  the network and so its potential influence was lost.
This novel approach has been first introduced  and experimentally evaluated in \cite{CordascoGR15}. Here we present a theoretical analysis of this  approach together
 with  a deeper experimental analysis. We will show that
although %it is true that 
the new heuristic is not always better than the  one in \cite{CGM+} (an example of such a rare case is   provided in Section \ref{undirectedGraph}), the novel approach has the following properties:
\begin{itemize}
\item[$\bullet$] It always produces an optimal solution for several classes of networks;
\item[$\bullet$] it always produces a solution $S$ of bounded cardinality matching
the upper bound given in  \cite{ABW-10} and \cite{CGM+}.
%, namely  $|S|\leq \sum_{v\in V} \min\left(1,\frac{t(v)}{d(v) +1}\right)$;
\end{itemize}
From a practical point of view, in real-life networks,  experiments show that: 
\begin{itemize}
\item[$\bullet$] The proposed algorithm produces solutions that always outperform the ones obtained using
the known algorithms which have a comparable running time   \cite{CGM+,SEP};
\item[$\bullet$]  the performance of  algorithms  for the Minimum Target Set  problem  correlates with the network modularity, which measures the strength of the  network subdivision  into communities. If
the modularity is high then  the influence  is hard to propagate among communities; on the other hand, when the network allows
the propagation of influence between different communities, the performance of the algorithms increases.  
This correlation becomes  stronger for the algorithm proposed in this paper. Such a  result
is probably due to the capability of our   algorithm  to better exploit situations where the community structure of the networks allows some  influence propagation between different communities.
\end{itemize}

\section{The Minimum Target Set Problem } \label{notation}\label{secMTS}
%In this section, we formally define the Minimum Target Set Problem.
%\medskip
We represent a social network by means of  a  directed  graph $G = (V,E)$
where an arc $(u,v)$ represents the capability of $u$ of influencing $v$. 
\\
A threshold  function  $t: V \to \N=\{0,1,2,\ldots\}$ assigns non negative integers  to the nodes of $G$:
 For each node $v\in V$, the value 
 $t(v)$  measures the conformity of node $v$, in the sense that an
 easy-to-conform  element $v$  of the network  has ``low''  threshold value $t(v)$ while 
a hard-to-conform  element $u$ has  ``high''  threshold value \cite{Gr}.
\\
We denote by $\Gamma_G^{\iin}(v)=\{u|\ (u,v)\in E\}$ and  by  
$\Gamma_G^{\oout}(v)=\{u|\ (v,u)\in E\}$, respectively, the incoming and outgoing neighborhood   of the node $v$ in $G$. 
Similarly, $d^{\iin}_G(v)=|\Gamma_G^{\iin}(v)|$ and $d^{\oout}_G(v)=|\Gamma_G^{\oout}(v)|$ 
denote the incoming and outgoing degree of the node $v$ in $G$.

When dealing with undirected graphs, as  usual, we     represent them  by the corresponding  \textit{bidirected} digraph where each  edge  is replaced by a pair  of opposite arcs.
In such a case, we    denote by  $d_G(v)=d_G^\iin(v)=d_G^\oout(v)$ 
the degree of  $v$ and by $\Gamma_G(v)=\Gamma_G^{\iin}(v)=\Gamma_G^{\oout}(v)$ the neighborhood   of   $v$ in $G$.
\\
Given a subset $V'\subseteq V$ of nodes of $G$, we denote by  $G[V']$  the subgraph of $G$  induced by nodes in  $V'$.  

Let   $G=(V,E)$  be a digraph with threshold function  $t: V \rightarrow \N$ and  $S\subseteq V$.
An {\em activation process in $G$ starting at $S$}
is a sequence\footnote{In the rest of the paper we will omit the subscript $G$  whenever the graph $G$  is clear from the context.}
$$\Active_G[S,0] \subseteq \Active_G[S,1] \subseteq \ldots\subseteq \Active_G[S,\ell] \subseteq \ldots \subseteq V$$
of node subsets, with $\Active_G[S,0] = S$ and,  for  $\ell \geq 1$
$$\Active_G[S,\ell]=\Active_G[S,\ell{-}1] \cup \Big\{u :\, \big|\Gamma_G^{\iin}(u)\cap \Active_G[S,\ell {-} 1]\big|\ge t(u)\Big\}.$$
In words, at each round  $\ell\geq 1$ the set of active nodes is
augmented by all the   nodes $u$ for which the  number of
\emph{already} active incoming neighbors  is at least  equal to
 $u$'s threshold $t(u)$. 
{The node $v$ is said to get  {\em  active} at round $\ell>0$ if $v \in  \Active[S,\ell]- \Active[S,\ell - 1]$.}

A {\em target set} for $G$ is a set $S\subseteq V$ such that $\Active_G[S,\lambda]=V$ for some  $\lambda\geq 0$. 
The  problem
we  study in this paper
 is defined as follows:
\begin{quote}
 {\sc MINIMUM TARGET SET (MTS)}.
 \\
{\bf Instance:} A digraph $G=(V,E)$, thresholds $t:V\rightarrow \N$.
 \\
{\bf Problem:} Find a target set $S\subseteq V$ of \emph{minimum} size for $G$.
\end{quote}

\section{The MTS algorithm}\label{sec:upper} 
We first present our algorithm  for the MTS problem.
The algorithm MTS($G$, $t$), given in Algorithm \ref{alg}, 
works by iteratively deprecating nodes from the input digraph unless a certain condition occurs which makes a node be added to the output target set.

We illustrate  the logic of  the algorithm MTS($G$, $t$) on the example  digraph $G$  in Fig. \ref{fig:example}(a). 
The number inside each circle represents the node threshold.
At each iteration, the algorithm selects a node and possibly deletes it from the graph.
Fig. \ref{fig:example} shows  the evolution of  $G$ (and of the  node thresholds) at the beginning  of each iteration of the algorithm. \\
The execution of the algorithm MTS  on the graph  in Fig. \ref{fig:example}(a) is described below and summarized in table \ref{example}.

\begin{algorithm}[H]
\SetCommentSty{footnotesize}
\SetKwInput{KwData}{Input}
%%\SetKwInput{KwResult}{Output}
\DontPrintSemicolon
\caption{ \ \   \textbf{Algorithm} MTS($G$, $t$) \label{alg}}
\KwData { A digraph $G=(V,E)$ with thresholds $t(v)$ for $v\in V$. }
%%\KwResult{ $S,$ a target set for $G.$}
%%\BlankLine
\setcounter{AlgoLine}{0}
$S=\emptyset$; \quad 
$L=\emptyset$; \quad 
$U=V$ \\  
\ForEach {$v\in V $}{
    $k( v)=t( v)$ \\
    $\d( v)=|\Gamma^{\iin}( v)|$\\ 
}
\While{ $U\neq \emptyset$ }{
    \eIf(\tcp*[f]{\underline{Case 1}:   $v$ gets active by the influence of its incoming neighbors in $V-U$ only; it can then influence its outgoing neighbors in $U$.}){there exists $v\in U$ s.t. $k(v)=0$}
    		{
    		\ForEach {$u\in \Gamma^{\oout}(v) \cap U$}{
    			$k(u)=\max(k(u)-1,0)$\\
    			\lIf{$v\notin L$}
    				{
    				$\d(u)=\d(u)-1$
    				}
    		}
    		$U=U-\{v\}$
    		}
    		{
				\eIf(\tcp*[f]{\underline{Case 2}:  $v$ is added to $S$, since no sufficient incoming neighbors remain in $U$ to activate it;  $v$  can then influence its outgoing neighbors in $U$.}){ there exists $v\in U{-}L$ s.t. $\d(v) <  k(v)$}
    				{$S=S\cup\{v\}$ \\
    				\ForEach {$u\in \Gamma^{\oout}(v) \cap U$}{
    						$k(u)=k(u)-1$\\
    						$\d(u)=\d(u)-1$
    				}
    				$U=U-\{v\}$
    				}
    				(\tcp*[f]{\underline{Case 3:} Node $v$ will  be  activated by its incoming neighbors in $U$. }){
    				$v={\tt argmax}_{u\in U-L}\left\{\frac{k(u)}{\d(u)(\d(u)+1)}\right\}$\\
    					\lForEach {$u\in \Gamma^{\oout}(v)\cap U$}{$\d(u)=\d(u)-1$	}
						$L=L\cup\{v\}$ \\
    				}
    		}
}
\textbf{return} $S$
\end{algorithm}

\smallskip
\noindent
The algorithm initializes  the target set $S$ to the empty set   and  a set $U$  (used to keep the surviving nodes of $G$) to $V$.
It then proceeds as follows:\\
{\bf Iteration 1.} \  If no node in $G$ has threshold either equal to $0$ or larger than the indegree, then  Case 3 of the algorithm  occurs and a node is selected according to the function at  line 19 of the algorithm. This function is based on the idea that nodes having low threshold and/or large degree are the less useful to start the activation process.
Both nodes $v_2$ and $v_3$ of the graph in Fig. \ref{fig:example}(a) satisfy the function, then the algorithm arbitrary chooses one of them\footnote{Notice that  in each of Cases 1, 2, and 3 ties are broken at random.}. 
Let $v_2$ be selected. Hence, $v_2$ is moved into a 
 \textit{limbo} state, represented by the set $L$. As a consequence
 of being in $L$, the outgoing neighbor $v_1$ of $v_2$ will not count on $v_2$ for being influenced (the value $\delta(v_1)$, which denotes the incoming degree of $v_1$ restricted to the nodes that belongs to the residual graph but not to  $L$, is reduced by 1). In Fig. \ref{fig:example}(b), the circle of $v_2$ and the arrow to its outgoing neighbor $v_1$ are dashed to represent this situation. \\
{\bf Iteration 2.} \   Due to this update, node $v_1$ in the residual digraph  in Fig. \ref{fig:example}(b)   remains with fewer ``\textit{usable}'' incoming neighbors than its  threshold (i.e., $\delta(v_1)=1<2=k(v_1)=t(v_1)$).
 Hence, at the second iteration Case 2 occurs  (note that no node has threshold equal to $0$) and $v_1$ is selected and added to the target set $S$. As a consequence, $v_1$ is deleted from $U$ (i.e., $v_1$ is removed from the residual digraph - see Fig. \ref{fig:example}(c)) and   the  thresholds of its outgoing neighbors are decreased by 1 (since they can receive $v_1$'s influence).\\
{\bf Iteration 3.} \   If  the residual digraph contains a node $v$ whose threshold has become 0 (e.g. the  nodes  which are   already in $S$ -- see Case 2 --  suffice  to activate $v$) then Case 1 occurs and the node  
$v$ is  selected and deleted from the digraph.
This case occurs  for the graph in Fig. \ref{fig:example}(c),  to nodes $v_2$ and $v_4$ with $k(v_2)=k(v_4)=0$. The algorithm arbitrary chooses one of them;   say $v_4$. Hence,  $v_4$ is selected and removed from $U$ (and  from the residual digraph) and the threshold of its outgoing neighbor $v_5$ is decreased by 1 (since once $v_4$ activates then   $v_5$ will receive its influence). See Fig. \ref{fig:example}(d).\\
{\bf Iteration 4.} \   Case 1  can also apply  to a node  $v\in L$. In such a case the value of $\delta(u)$, for each   outgoing neighbor $u$  of the selected node $v$, were already  reduced by $1$---when $v$ was added  to $L$---and, therefore, it is not reduced further. 
\\
In our example, at the fourth iteration $k(v_2)=0$  and  Case 1 occurs. Hence, $v_2$ is selected and deleted from $U$ - see Fig. \ref{fig:example}(e).
%We point out that by construction, once a node is moved to $L$, then it will be removed from the graph only in   Case 1; indeed, Case 2 and 3   only apply to nodes outside $L$. In other words, nodes moved to  $L$ will never belong to the target set.\\
%
\\
{\bf Iteration 5.} \  Now, Case 3 occurs. Both nodes $v_3$ and $v_5$ maximize the   function at line 19. 
Let $v_3$ be the selected node.
Hence, $v_3$ is added to $L$ and all its outgoing neighbors, $v_5$ and $v_6$,  have the $\delta()$ value  reduced by 1. See Fig.\ref{fig:example}(f).\\
{\bf Iteration 6.} \   Case 2 occurs since 
$\delta(v_6)=1<2=k(v_6)$. Hence, $v_6$ is selected and added to $S$ (since no sufficient incoming neighbors remain in the residual digraph to activate it),
the threshold of its outgoing neighbor $v_3$ is decreased by 1 (i.e., $v_3$ is influenced by $v_6$) and $v_6$ is removed from $U$ and so from the residual digraph. See Fig. \ref{fig:example}(g).\\
{\bf Iteration 7.} \   Case 1 occurs since $k(v_3)=0$.
Hence, $v_3$ is selected and removed from $U$.
Its outgoing neighbor $v_5$ has the threshold decreased by 1, since it receive the influence of $v_3$ that can be considered active. See Fig. \ref{fig:example}(h).\\
{\bf Iteration 8.} \  Finally, the last node $v_5$ in the residual digraph is selected (i.e., Case 1 occurs) and removed. The set $U$ is now empty and the algorithm stops returning the target set $S$.

\begin{remk}
We notice that if a node is added to the set  $L$, it  will never belong to the target set.
Indeed a node $v$ is  added to $S$ only if Case 2 occurs for $v$.
However,  Case 2 is restricted to nodes outside $L$  (see line 12 of the algorithm). \\
Hence, the condition of the while loop at line 5 could be changed to $U-L \neq \emptyset$ thus shortening the execution of the algorithm. 
We decided to use the  $U \neq \emptyset$ condition   because this results in  simplified  proofs without affecting the theoretical upper bound on the running time.
\end{remk}

{
\begin{center}
\begin{figure}[th!]

%\vspace*{-0.1truecm}
		\includegraphics[height=11.7truecm,width=12.0truecm]{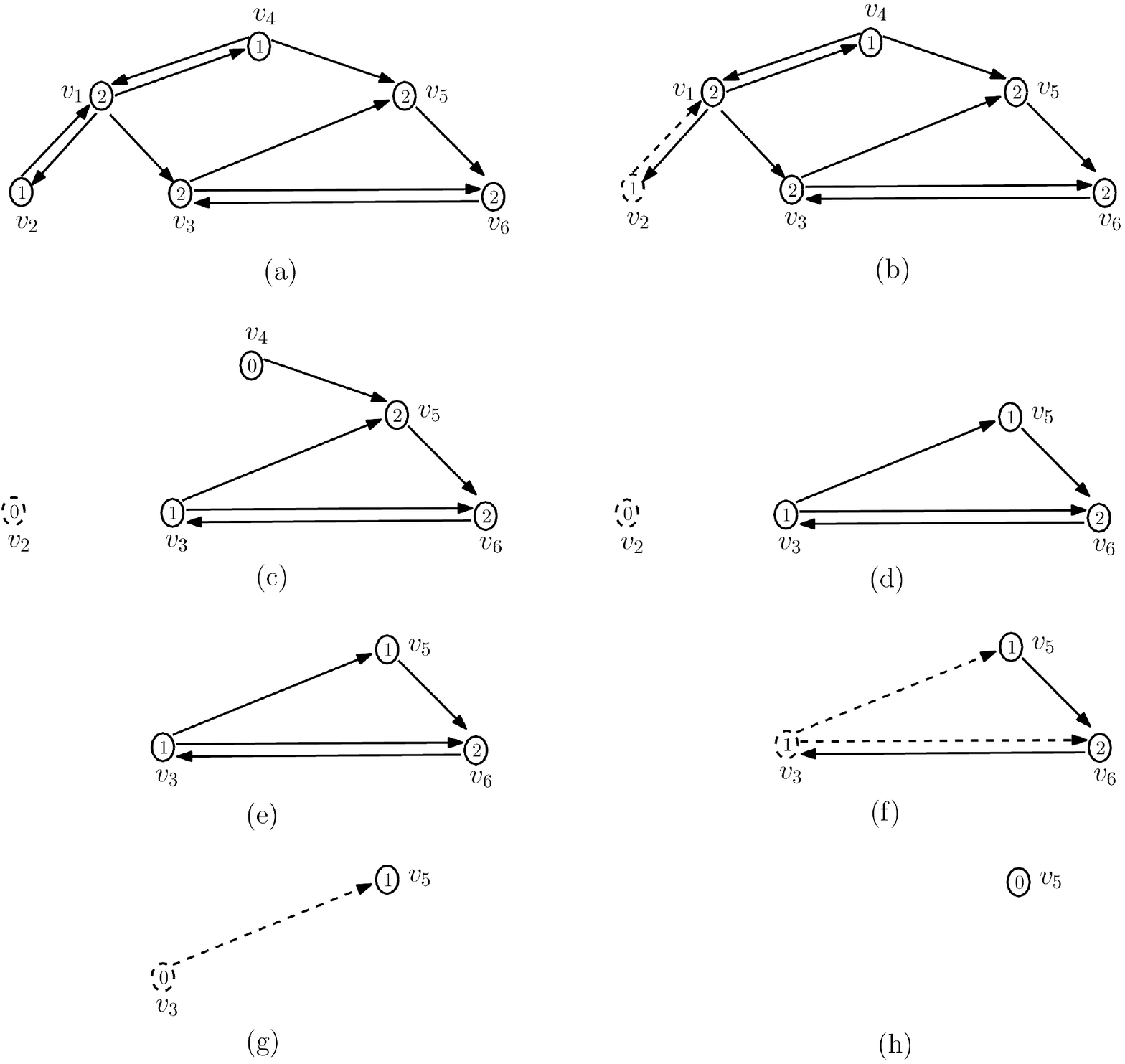}
		\vspace*{-0.5truecm}
		\caption{	The evolution of a digraph $G$ during the execution of algorithm MTS($G,t$); dashed circles and dashed arrows represent nodes moved in the set $L$ and their outgoing arcs, respectively. The values inside circles represent the residual thresholds.	\label{fig:example}}
\end{figure}

\smallskip

\begin{table}[ht!]
%\resizebox{\linewidth}{!} {
\begin{center}
\begin{tabular}{|c|c|c|c||c|c|}
\hline
i 		 & $U$ & $L {\cap} U$ & $S$ & Selected node	& Case \\ \hline
1 		& $\{v_1,v_2,v_3,v_4,v_5,v_6\}$ & $\emptyset$ 	& $\emptyset$ & $v_2$	& 3   \\ \hline
2 		& $\{v_1,v_2,v_3,v_4,v_5,v_6\}$ & $v_2$ 	& $\emptyset$ & $v_1$	& 2   \\ \hline
3 		& $\{v_2,v_3,v_4,v_5,v_6\}$     & $v_2$ 	& $v_1$ & $v_4$	& 1   \\ \hline
4 		& $\{v_2,v_3,v_5,v_6\}$         & $v_2$ 	& $v_1$ & $v_2$	& 1   \\ \hline
5 		& $\{v_3,v_5,v_6\}$             & $\emptyset$ 	& $v_1$ & $v_3$	& 3   \\ \hline
6 		& $\{v_3,v_5,v_6\}$             & $v_3$ 	& $v_1$ & $v_6$	& 2   \\ \hline
7 		& $\{v_3,v_5\}$                 & $v_3$ 	& $v_1,v_6$ & $v_3$	& 1   \\ \hline
8 		& $\{v_5\}$                     & $\emptyset$ 	& $v_1,v_6$ & $v_5$	& 1   \\ \hline
%9 		& $\emptyset$                    & $\emptyset$ 	& $v_1,v_6$ & 	&    \\ \hline
\end{tabular}
%}
\end{center}
\caption{	The execution of the algorithm MTS($G$, $t$) for the graph in Fig. \ref{fig:example}(a). 
For each iteration of the while loop, the tables provides the content of  the sets $U$, $L {\cap} U$, $S$ at the begin of iteration, the selected node and whether Cases 1, 2  or 3 applies.\label{example} }
\end{table}
\end{center}}

\subsection{Algorithm Correctness}
{{We  prove now  that the proposed algorithm always outputs a target set for the input graph and evaluate its running time.
To this aim, we first introduce some notation  and properties of the algorithm MTS that will be used in the sequel of the paper. }}

We denote by $n$ the number of nodes in $G$, that is   $n=|V|$, and by $\lambda$ the number of iterations of the while loop of algorithm MTS($G$, $t$).  Moreover, we denote: 
\begin{itemize}
\item by $v_i$ the node that is selected during the $i$-th iteration of the while loop in MTS($G$, $t$), for $i=1,\ldots,\lambda$;
\item by  $U_i, L_i, S_i, \d_i(u),$ and $k_i(u)$, the sets $U, L, S$ and the values of $\d(u), k(u)$, respectively, as updated at the beginning of  the $i$-th iteration of the while loop in MTS($G$, $t$).
\end{itemize}
For the initial value $i=1$,  the above values are those of the input graph $G$, that is: 
$U_1=V$, $G[U_1]=G$ and  $\delta_1(v)=d^{in}(v)$,  $k_1(v)=t(v)$, for each  $v$ in $G$.
\\
%Notice that $\delta_i(u)$ counts the number of  arcs incoming in  $u$  from nodes in $U_i-L_i$, hence   
%
%
The properties  stated below will be useful in the rest of the paper.
\begin{fact} \label{fact1}
For each iteration $i$ of the while loop in MTS($G$, $t$) and for each $u \in U_i$, 
$$\delta_i(u)=
|\Gamma^{\iin}(u) \cap (U_i-L_i)|\leq d^{\iin}_{G[U_i]}(u).$$ 
Furthermore, if $G$ is bidirectional then $$\delta_i(u)=|\Gamma^{\iin}(u) \cap (U_i-L_i)|=|\Gamma^{\oout}(u) \cap (U_i-L_i)|.$$
\end{fact}

\begin{fact} \label{fact2}
For each iteration $1\leq i< \lambda$, let $v_i$ be  the node that is selected during the $i$-th iteration of the while loop in MTS($G$, $t$). We have that $$
U_{i+1} - L_{i+1}=\begin{cases} U_i-L_i & \mbox{ if  $v_i \in L_i$;} \\ U_i-L_i - \{v_i\} & \mbox { otherwise.}\end{cases}
$$
\end{fact}

\smallskip
The following Lemma establishes an upper bound on the number of iterations of the while loop of the algorithm.
This result, a part telling us that the algorithm ends on any input graph, will be useful for the running time evaluation. 

\begin{lemma} \label{lemma1}
The number of iterations of the while loop of algorithm MTS($G$, $t$) is at most $2n$ (i.e., $\la\leq 2n$).
\end{lemma}
\begin{proof}{}
First of all we prove that, at each iteration $i \geq 1$ of the while loop of MTS($G$, $t$), a node  $v_i \in U_i$ is selected.
If $U_i-L_i \neq \emptyset$ then there obviously exists a node $u \in U_i$ for which one of the three cases in the while loop of MTS holds.
Now, we show that for any $i \geq 1,$ it holds
\begin{equation} \label{U=L}
\mbox{\textbf{If} } U_i-L_i=\emptyset \mbox{ \textbf{then}   there exists} \  u \in U_i  \ \mbox{with} \
k_i(u){=}0.   
\end{equation}
Assume that there exists an iteration $i \geq 1$ such that $U_i-L_i=\emptyset$. 
Let $u$ be the node, among the nodes in $L_i$, that is inserted last in $L$ at some iteration $j<i$.
Since Case 3 holds for $u$ at iteration $j$ we have $0 < k_j(u) \leq  \d_j(u)$. 
As a consequence, all the nodes eventually selected at iterations $j+1,\ldots,i-1,$ are nodes for which either Case  1 or Case 2 holds.
Since by the algorithm  once a node is moved into the set $L$, the value $\d$ of each of its outgoing neighbors is decreased by $1$ (cfr. lines 21--22),
we have that $$|U_j - L_j| \geq \d_j(u).$$
%$i < j-1$ 
Recalling that $\d_j(u) \geq k_j(u) >0$ we get that at least $k_j(u)$ among the incoming neighbors of $u$ in $G[U_j]$ are selected during iterations $j+1, \ldots, i-1$ and for each of them the residual threshold of $u$ is decreased by one (cfr. lines 8 and 16). This leads to $k_i(u)=0$ and (\ref{U=L}) holds. 

We conclude the proof noticing each  $v \in V$ can be selected at most twice: Once $v$ is eventually 
inserted in $L$ (if Case 3 applies) and once $v$ is removed from $U$ 
(if either Case 1 or Case 2 apply). \qed 
\end{proof}

\smallskip

{{We are now ready to prove the correctness of the proposed algorithm, namely that the algorithm MTS($G$, $t$) always returns a target set for the input digraph with the given thresholds. }}

\begin{theorem}\label{teo1}
For any graph $G$ and threshold function $t$, the algorithm MTS($G$, $t$) outputs a target set for 
$G$.
\end{theorem}
\begin{proof}{}
%Let $\la$ be  the number of iterations of the while loop in MTS($G$). 
%Furthermore, for each $i=1,\ldots,\la$ let
%\begin{itemize}
%\item $v_i$ be the node that is selected during the $i$-th iteration; 
%\item $S_i$ be the set $S$ as updated at the end of the $i$-th iteration;
%\item  $G[U_i]$ be the subgraph of $G$ induced by $U_i$.
%\end{itemize}
%%digraph with $A_i = \{ (v,u) \ | \ u \in \Nio(v)\}$. 
%Note that $G(1)=G$, $N_1^{\oout}(v)=\Gamma_G^{\oout}(v)$, $N_1^{\iin}(v)=\Gamma_G^{\iin}(v)$, $k_1(v)=t(v)$, for each node $v$ of $G$.
%Also, note that for each $v \in L_i \cap U_i$  and $u \in \Nio(v)$ we have 
%$v \in \Gamma_{G[U_i]}^{\iin}(u)$ while $v \not \in \Nin(u)$ (cfr. line 29).
%\\
We  show that for each $i=1,\ldots,\la$ the set 
$S_i$ is a target set for the digraph $G[U_i]$, 
assuming that   each node $u$ in $G[U_i]$ has  threshold $k_i(u)$.
The proof is by induction on  $i$, with $i$ going from $\la$ down to $1$.
\\
\remove{%BEGIN-REMOVE
??????
\\
Consider first
 $i=\la$.  Then either  $v_\la\in U_\la$  has threshold $k_\la(v_\la)=0$  and 
$S_\la=\emptyset$ or
the unique node $v_\la\in U_\la-L_\la$  has positive threshold  $k_\la(v_\la)>\d_\la(v_\la)=0$ and   $S_\la=\{v_\la\}$.
By Lemma \ref{lemma1bis}, this is correct.
\\
?????
} %END-REMOVE
Consider first $i=\la$.   
The unique node $v_\la$ in $G(\la)$  either has threshold $k_\la(v_\la)=0$  and $S_\la=\emptyset$ or the node has positive threshold  $k_\la(v_\la)>\d_\la(v_\la)=0$ and   $S_\la=\{v_\la\}$.

\noindent
Consider now $i<\la$ and suppose the algorithm be correct on $G[U_{i+1}]$, that is, 
$S_{i+1}$ is a target set for $G[U_{i+1}]$ with thresholds  $k_{i+1}(u)$ for $u\in U_{i+1}$.
We show that the algorithm is correct on $G[U_i]$ with thresholds 
$k_{i}(u)$ for $u\in U_{i}$.
\\
By the algorithm MTS, for each $u \in U_i$ we have
\begin{equation} 
k_{i+1}(u) {=} 
\begin{cases} \label{ki+1}
\max(k_i(u) {-} 1, 0) & \mbox {if } u \in \Gamma^{out}(v_i) \cap U_i \mbox{ and   } 
                              (k_i(v_i){=}0 \mbox{ or } k_i(v_i)> \d_i(v_i))\\
k_i(u) & \mbox{otherwise.}
\end{cases}
\end{equation}
We distinguish three cases on the selected node $v_i$.
\begin{itemize}
\item  $1 \leq k_i(v_i) \leq \d_i(v_i)$ (Case 3 holds). 
In this case $U_i=U_{i+1}$. Moreover by (\ref{ki+1}),  $k_{i+1}(u)=k_i(u)$ for $u\in U_{i+1}$. Hence 
the target set $S_{i+1}=S_{i}$ for $G[U_{i+1}]$ is also a target set for $G[U_i]$.

\item  $k_i(v_i) > \d_i(v_i)$ (Case 2 holds). In this case $U_{i+1}=U_{i}-\{v_i\}$
and $S_i = S_{i+1} \cup \{v_i\}$. By  (\ref{ki+1}) 
it follows that for any $\ell \geq 0$, 
$$\Active_{G[U_i]}[S_{i+1} \cup \{v_i\}, \ell] - \{v_i\} = \Active_{G[U_{i+1}]}[S_{i+1}, \ell].$$
Hence, $\Active_{G[U_i]}[S_i, \ell] = \Active_{G[U_{i+1}]}[S_{i+1}, \ell] \cup \{v_i\}$
and $S_i$ is a target set for $G[U_i]$.

\item $k_i(v_i) = 0$ (Case 1 holds). 
Since $k_i(v_i) = 0$, node $v_i$ is immediately active in $G[U_i]$. 
Hence by  (\ref{ki+1}),  each outgoing neighbor  
$u$ of $v_i$ in $G[U_i]$ is influenced by $v_i$ and its threshold is updated 
according to (\ref{ki+1}).
Therefore, since $S_{i+1}$ is a target set  for $G[U_{i+1}]$, we have that $S_i=S_{i+1}$ is  a target set for $G[U_i]$.
\end{itemize}
The theorem follows since $G[U_1]=G$.  
\end{proof}
\subsection{ Running Time}
 The MTS algorithm can be implemented  to run in 
$O(|E| \log |V |)$ time. Indeed we need to process the nodes $v \in V$--each one at most two times (see Lemma \ref{lemma1})--according to
the metric $k(v) /(\d(v)(\d(v)+1))$, and the updates, that follows each processed node
$v \in V$ involve at most $d^{\oout}(v)$ outgoing neighbors of $v$.\\
{{It is worth to mention that the MTS algorithm running time is comparable with the running time of the state of the art strategies for the  MTS problem. Indeed, also these strategies usually need to sort the nodes according to some  metric and  to keep them sorted  after a change in the graph}.}
\section{Undirected graphs} \label{undirectedGraph}
Recall that here $\Gamma(v)=\Gamma^{\iin}(v)=\Gamma^{\oout}(v)$ and $\delta_i(v)=|\Gamma(v) \cap (U_i-L_i)|$ for each $v \in U_i$ and $i=1, \ldots \lambda$.
\subsection{Optimality on Trees, Cycles and Cliques}\label{sec:trees} 
The main result of this section is the following Theorem.	
\begin{theorem} \label{th2}
The algorithm  MTS($G$, $t$) returns an optimal solution   whenever the input graph is either a tree, a cycle, or a clique.
\end{theorem}		
In Theorem \ref{optimal} we will prove that our  MTS algorithm provides an optimal solution for a family of graphs whenever the  
TSS algorithm, designed in \cite{CGM+} and shown   in Algorithm \ref{algTSS}, does. This and the results in  \cite{CGM+}  imply, in particular, the optimality of the MTS algorithm in case of 
trees, cycles and cliques.

{The MTS algorithm is an improvement of the TSS algorithm.
The main difference between the two algorithms is that the MTS algorithm takes also into account the potential influence that a deprecated node (i.e., a node selected in Case 3) may apply on his outgoing neighbors. For this reasons such nodes are moved into a limbo state (that is the node has been discarded but not removed) while in the original TSS
algorithm, once a node was selected in Case 3, it was immediately pruned from the graph and so its potential influence was lost.}

Even though the MTS algorithm usually performs   better than the TSS algorithm---as it is also shown  by the experiments in the Section \ref{optimal} ---the following example gives a rare case in which the  TSS algorithm outperforms the  MTS algorithm.
\begin{example}
{Consider the graph $G$  in Fig. \ref{fig:exampleC}. The number inside each circle is the node threshold.
A possible execution of the two algorithms MTS and TSS on $G$ is summarized in tables \ref{example1} and \ref{example2}. 
For each iteration $i$ of the while loop, the tables provides the content of  the sets $U_i$, $L_i {\cap} U_i$, $S_i$, the selected node and whether Cases 1, 2  or 3 applies.
Analyzing the tables one can observe that the algorithm TSS provides a target set of cardinality $2$ (which is optimal in this case) while the algorithm MTS provides a target set of cardinality $3$. It is worth to mention that the two algorithms performs very similarly. For instance the two graphs obtained at the begin of round $4$ of the MTS algorithm and round $3$ of the TSS algorithm are identical except for the threshold of the node $v_2$ which in the MTS algorithm is reduced to $3$ thanks to the contribution of the node $v_1$ (see Fig. \ref{fig:exampleC}(bottom-left)). Indeed node $v_1$ is first placed in $L$ at iteration $1$ and it is removed from the graph  at iteration $3$ as  its residual threshold becomes $0$. In the TSS algorithm, node $v_1$ is directly removed from the graph at iteration $1$ and the threshold of $v_2$ remains $4$ (see Fig. \ref{fig:exampleC}(bottom-right)). }

\smallskip

\begin{algorithm}[H]
\SetCommentSty{footnotesize}
\SetKwInput{KwData}{Input}
\SetKwInput{KwResult}{Output}
\DontPrintSemicolon
\caption{ \ \   \textbf{Algorithm} TSS($G$, $t$) \cite{CGM+} \label{algTSS}}
\KwData { A graph $G=(V,E)$ with thresholds $t(v)$  for $v\in V$.\\ }
%\KwResult{$S,$ a target set for $G$.}
%\BlankLine
\setcounter{AlgoLine}{0}
$S=\emptyset$; $U=V$; \\  
\ForEach {$v\in V $}{
  $k( v)=t( v)$; \\ 
  $\d( v)=|\Gamma(v)|$;\\
}
\While(\tcp*[f]{Select one node and eliminate it from the graph.}){ $U\neq \emptyset$ }{
    \eIf(\tcp*[f]{\underline{Case 1.}}){there exists $v\in U$ s.t. $k(v)=0$}
    		{
    		\lForEach {$u\in \Gamma(v)\cap U$}{
    			$k(u)=\max(k(u)-1,0)$;
    		  }
    		}
				{
				\eIf(\tcp*[f]{\underline{Case 2.}}){ there exists $v\in U$ s.t. $\d(v) <  k(v)$}
    				{$S=S\cup\{v\}$; \\
    				\lForEach {$u\in \Gamma(v)\cap U$}{
    						$k(u)=k(u)-1$;
    				}
						}
						(\tcp*[f]{\underline{Case 3.}}){	
    				$v={\tt argmax}_{u\in U}\left\{\frac{ k(u)}{\delta(u)(\delta(u)+1)}\right\}$;
    				}
						}
						\lForEach(\tcp*[f]{Remove  $v$.}) {$u\in \Gamma(v)\cap U$}{
    						$\d(u)=\d(u)-1$;
							%\\	$N(u)=N(u)-\{v\}$;
    				}
						$U=U-\{v\}$;
}
\textbf{return} $S$
\end{algorithm}

\begin{figure}[th!]
\begin{center}

	\caption{	A example of graph where the TSS algorithm provides a better solution. (Top) the initial graph $G$. (Bottom-left) The residual graph at the beginning of round $4$ of the algorithm MTS. (Bottom-right) The residual graph at the beginning of round $3$ for of the algorithm TSS.	\label{fig:exampleC}}
	\end{center}
\end{figure}

\begin{table}[ht!]
\begin{center}
\scalebox{0.95}{
\resizebox{1.05\linewidth}{!} {
\begin{tabular}{|c|c|c|c||c|c|}
\hline
$i$		 & $U_i$ & $L_i {\cap} U_i$ & $S_i$ & Selected	& Case \\
	 &  &  & &  node 	& \\ \hline
1 		   & $\{v_0,v_1,v_2,v_3,v_4,v_5,v_6,v_7,v_8,v_9,v_{10}\}$ &  $\emptyset$	   	& $\emptyset$ & $v_1$	& 3 \\ \hline
2 		 & $\{v_0,v_1,v_2,v_3,v_4,v_5,v_6,v_7,v_8,v_9,v_{10}\}$ & $\{v_1\}$ 	   	& $\emptyset$ & $v_0$	& 2 \\ \hline
3 		& $\{v_1,v_2,v_3,v_4,v_5,v_6,v_7,v_8,v_9,v_{10}\}$ 		& $\{v_1\}$ 			& $\{v_0\}$  & $v_1$	& 1 \\ \hline
4 		 & $\{v_2,v_3,v_4,v_5,v_6,v_7,v_8,v_9,v_{10}\}$  				& $\emptyset$  			& $\{v_0\}$ & $v_{10}$	& 3 \\ \hline
5 		& $\{v_2,v_3,v_4,v_5,v_6,v_7,v_8,v_9,v_{10}\}$ 				& $\{v_{10}\}$ 			& $\{v_0\}$  & $v_2$	& 3 \\ \hline
6 		 & $\{v_2,v_3,v_4,v_5,v_6,v_7,v_8,v_9,v_{10}\}$ 				& $\{v_2,v_{10}\}$ 			& $\{v_0\}$  & $v_5$	& 3 \\ \hline
7 		 & $\{v_2,v_3,v_4,v_5,v_6,v_7,v_8,v_9,v_{10}\}$ 				& $\{v_2,v_5,v_{10}\}$ 			& $\{v_0\}$ & $v_9$	& 3\\ \hline
8 		 & $\{v_2,v_3,v_4,v_5,v_6,v_7,v_8,v_9,v_{10}\}$ 				& $\{v_2,v_5,v_9,v_{10}\}$ 			& $\{v_0\}$ & $v_4$	&2  \\ \hline
9 		  & $\{v_2,v_3,v_5,v_6,v_7,v_8,v_9,v_{10}\}$ 				& $\{v_2,v_5,v_9,v_{10}\}$ 			& $\{v_0,v_4\}$ & $v_6$	& 1\\ \hline
10 		 & $\{v_2,v_3,v_5,v_7,v_8,v_9,v_{10}\}$ 				& $\{v_2,v_5,v_9,v_{10}\}$ 			& $\{v_0,v_4\}$ & $v_7$	& 1 \\ \hline
11 		 & $\{v_2,v_3,v_5,v_8,v_9,v_{10}\}$ 				& $\{v_2,v_5,v_9,v_{10}\}$ 			& $\{v_0,v_4\}$  & $v_9$	& 1\\ \hline
12 		 & $\{v_2,v_3,v_5,v_8,v_{10}\}$ 				& $\{v_2,v_5,v_{10}\}$ 			& $\{v_0,v_4\}$  & $v_3$	& 3\\ \hline
13 		 & $\{v_2,v_3,v_5,v_8,v_{10}\}$ 				& $\{v_2,v_3,v_5,v_{10}\}$ 			& $\{v_0,v_4\}$ & $v_8$	& 2 \\ \hline
%13 		& $v_2$	& 1 & $\{v_3,v_5,v_8,v_{10}\}$ 				& $\{v_3,v_5,v_{10}\}$ 			& $\{v_0,v_4\}$  \\ \hline
14 		 & $\{v_2,v_3,v_5,v_{10}\}$ 				& $\{v_2,v_3,v_5,v_{10}\}$ 			& $\{v_0,v_4,v_8\}$ & $v_3$	& 1 \\ \hline
15 		 & $\{v_2,v_5,v_{10}\}$ 				& $\{v_2,v_5,v_{10}\}$ 			& $\{v_0,v_4,v_8\}$ & $v_5$	& 1 \\ \hline
16 		& $\{v_2,v_{10}\}$ 				& $\{v_2,v_{10}\}$ 			& $\{v_0,v_4,v_8\}$  & $v_2$	& 1 \\ \hline
17 		 & $\{v_{10}\}$ 				& $\{v_{10}\}$ 			& $\{v_0,v_4,v_8\}$ & $v_{10}$	& 1 \\ \hline
%18 		& $\emptyset$ 				& $\emptyset$ 			& $\{v_0,v_4,v_8\}$  & & \\ \hline
\end{tabular}
}
}
\caption{	An example of execution of  MTS($G$, $t$) on the graph $G$  in Fig. \ref{fig:exampleC} (top).\label{example1} }
%\end{center}
\end{center}
\end{table}
\begin{table}[ht!]
\begin{center}
%\resizebox{\linewidth}{!} {
\scalebox{0.90}{
\begin{tabular}{|c|c|c||c|c|}
\hline
$i$		 & $U_i$  & $S_i$ & Selected  	& Case \\ 	 &  &  &  node 	&  \\ \hline
1 		& $\{v_0,v_1,v_2,v_3,v_4,v_5,v_6,v_7,v_8,v_9,v_{10}\}$ 	   	& $\emptyset$ & $v_1$	& 3  \\ \hline
2 		& $\{v_0,v_2,v_3,v_4,v_5,v_6,v_7,v_8,v_9,v_{10}\}$ 	   	& $\emptyset$ & $v_0$	& 2 \\ \hline
3 		 & $\{v_2,v_3,v_4,v_5,v_6,v_7,v_8,v_9,v_{10}\}$ 	   	& $\{v_0\}$ & $v_2$	& 3 \\ \hline
4 		 & $\{v_3,v_4,v_5,v_6,v_7,v_8,v_9,v_{10}\}$ 	   	& $\{v_0\}$ & $v_{10}$	& 2 \\ \hline
5 		 &$\{v_3,v_4,v_5,v_6,v_7,v_8,v_9\}$ 	   	& $\{v_0,v_{10}\}$ & $v_9$	& 1 \\ \hline
6 		 & $\{v_3,v_4,v_5,v_6,v_7,v_8\}$ 	   	& $\{v_0,v_{10}\}$ & $v_8$	& 1 \\ \hline
7 		 & $\{v_3,v_4,v_5,v_6,v_7\}$ 	   	& $\{v_0,v_{10}\}$  & $v_7$	& 1\\ \hline
8 		 &  $\{v_3,v_4,v_5,v_6\}$ 	   	& $\{v_0,v_{10}\}$  & $v_6$	& 1\\ \hline
9 		 & $\{v_3,v_4,v_5\}$ 	   	& $\{v_0,v_{10}\}$ & $v_5$	& 1 \\ \hline
10 		 & $\{v_3,v_4\}$ 	   	& $\{v_0,v_{10}\}$ & $v_4$	& 1 \\ \hline
11 		 & $\{v_3\}$ 	   	& $\{v_0,v_{10}\}$ & $v_3$	& 1 \\ \hline
%12 		 & $\emptyset$ 			& $\{v_0,v_{10}\}$  & & \\ \hline
\end{tabular}
}
\caption{	An example of execution of  TSS($G$, $t$) on the graph $G$  in Fig. \ref{fig:exampleC} (top).\label{example2} }
%\end{center}
\end{center}
	\end{table}
\end{example}

In order to prove Theorem \ref{optimal}, we first  need an intermediate result.
\begin{lemma} \label{lemmaOptimal}
Let $G=(V,E)$ be a graph. Let $S^{*}(G,t)$ denote  an optimal target set for  $G$ with   threshold function $t$. Then for every pair of  functions $t_1$ and $t_2$ such that $t_1(v) \leq t_2(v)$ for each $v \in V$, 
it holds $|S^{*}(G,t_1)|\leq |S^{*}(G,t_2)|$.
\end{lemma}
\proof
It is enough to observe that since $t_1(v) \leq t_2(v)$ for each $v \in V$, the target set $S^{*}(G,t_2)$ for $G$ with threshold function $t_2$ is also a target set for $G$ with threshold function $t_1$. % and $S^{*}$ is optimal.
\qed

\medskip
 A graph family  is called  {\em hereditary} if   it is closed under induced subgraphs.
Let $\cal{G}$  be a hereditary graph family.
We say that an algorithm is optimal for  $\cal{G}$ if 
it  returns an optimal target set   for any  $G\in \cal{G}$ (for any threshold function on the nodes of $G$).

\newcommand\Sol[3]{{S_{#1}(#2,#3)}}

\begin{theorem} \label{optimal}
Let $\cal{G}$ be any hereditary family of graphs. If the algorithm TSS  is optimal for  
$\cal{G},$ then  the algorithm  MTS is optimal for   $\cal{G}$.
\end{theorem}

\proof
Let $G=(V,E)\in \cal{G}$ and $t:V\to \N$.
We recall that $\lambda$ denotes  the last iteration of algorithm MTS($G$, $t$) and that
for $i=1,\ldots,\lambda$:
\begin{itemize}
	\item $v_i$ denotes the node that is selected during the $i$-th iteration of the while loop in MTS($G$, $t$);
  \item  $U_i, L_i, S_i, \d_i(u),$ and $k_i(u)$, respectively, denote the sets $U, L, S$ and the values of $\d(u)$ and  $k(u)$,  as updated at the beginning of  the $i$-th iteration of the while loop in MTS($G$, $t$).
\end{itemize}
Moreover, we denote by $\Sol{Alg}{G}{t}$ the solution obtained by the algorithm $Alg$ when the input  is  $G$ with threshold function  $t$.
We prove that 
$$| \Sol{MTS}{G}{t} |=| \Sol{TSS}{G}{t}|.$$
We prove  that at  any  iteration  $i$ of  the while loop in MTS($G$, $t$) such that   $v_i \notin L_i$  it holds
 $|\Sol{MTS}{G[U_i]}{k_i} |=|\Sol{TSS}{G[U_i-L_i]}{k_i} |$. 
\\
Let $\lambda' \leq \lambda$ be the last iteration of the while loop in MTS($G$, $t$) for which   $v_{\lambda'} \notin L_{\lambda'}$. 
The proof proceeds by induction on $i$ going from $\lambda'$ down to $1$.
The theorem follows since for $i=1$ we get $L_1=\emptyset$
 (and therefore $v_1\not\in L_1$), $U_1=U_1-L_1=V$, 
 and $k_1(v)=t(v)$ for each node $v\in V$.\\

Since $U_{\lambda'}-L_{\lambda'}=\{v_{\lambda'}\}$, 
by Fact \ref{fact1}  we get $\delta_{\lambda'}(v_{\lambda'})=|\Gamma(v_{\lambda'}) \cap (U_{\lambda'}-L_{\lambda'})|=0$. Hence, by the algorithm and recalling Remark 1 we have
$$    \Sol{MTS}{G[U_{\lambda'}]}{k_{\lambda'}} =\begin{cases} 
\emptyset & \mbox{if $k_{\lambda'}(v_{\lambda'})=0$;} \\ 
\{v_{\lambda'} \} & \mbox{otherwise,} \end{cases} $$
that matches $\Sol{TSS}{G[U_{\lambda'}-L_{\lambda'}]}{k_{\lambda'}}$.

Let  $\ell$ be any iteration of the while loop in MTS($G$, $t$) such that 
$v_\ell \notin L_\ell $. Assume that 
\begin{equation}
|\Sol{MTS}{G[U_\ell]}{k_{\ell}}|=| \Sol{TSS}{G[U_\ell-L_\ell]}{k_{\ell}}|. \label{eq-x}
\end{equation}
Let $i \in\{1,\ldots, \ell-1\}$ be the unique  iteration for which 
$v_{i} \notin L_{i} $ and $v_j \in L_{j}$ for each $j=i+1,\ldots, \ell-1$.
\\
By the algorithm MTS we have that 

a) $k(v_{i+1})= \ldots =k(v_{\ell-1})=0$, 

b) $U_{i}-L_{i}-\{v_{i}\} = U_{i+1}-L_{i+1}=\ldots=U_{\ell}-L_{\ell}  $, 

c)  $k_{\ell}(v)\leq k_{i+1}(v)$, for all $v \in U_{\ell}-L_{\ell} $.

\medskip

\noindent
By a), Case 1 of the algorithm MTS occurs at each of the iterations 
from $i+1$ to $\ell-1$. As a consequence,   we clearly have 
\begin{equation}\label{eq-opt1}
\Sol{MTS}{G[U_{i+1}]}{k_{i+1}} =\ldots = \Sol{MTS}{G[U_{\ell}]}{k_{\ell}}.
\end{equation}
By (\ref{eq-opt1}), (\ref{eq-x}),  b), c) and Lemma  \ref{lemmaOptimal} (in this specific order), we get  
\begin{eqnarray}  \nonumber
|\Sol{MTS}{G[U_{i+1}]}{k_{i+1}}|&=&
|\Sol{MTS}{G[U_{\ell}]}{k_{\ell}}|\\ \nonumber
             &=& |\Sol{TSS}{G[U_{\ell}-L_\ell]}{k_{\ell}}|\\ \nonumber
	    &=& |\Sol{TSS}{G[U_{i+1}-L_{i+1}]}{k_{\ell}}|\\  \label{eq-opt2}
	  &\leq& |\Sol{TSS}{G[U_{i+1}-L_{i+1}]}{k_{i+1}}|. 
\end{eqnarray}
We now notice that if the algorithm MTS($G[U_{i}]$, $k_{i}$) adds the node $v_i$ to the target set, it does so because $k_{i}(v_{i}) > \d_{i}(v_{i})$. In this case, it is not difficult to see that there exists an execution of the algorithm TSS($G[U_{i}-L_{i}]$, $k_{i}$) that similarly adds the node $v_i$ to the target set. Therefore, if 
$$|\Sol{MTS}{G[U_{i}]}{k_{i}}|=|\Sol{MTS}{G[U_{i+1}]} {k_{i+1}}| +1 $$  then  $$|\Sol{TSS}{G[U_{i}-L_i]}{k_{i}}|=|\Sol{TSS}{G[U_{i+1}-L_{i+1}]} {k_{i+1}}| +1. $$
%
%$$|\mbox{MTS($G[U_{i}]$, $k_{i}$)}|=\begin{cases} 
%|\mbox{MTS($G[U_{i+1}]$, $k_{i+1}$)}| +1 & \mbox{if $k_{i}(v_{i}) > \d_{i}(v_{i})$,} \\
%|\mbox{MTS($G[U_{i+1}]$, $k_{i+1}$)}| & \mbox{otherwise,} 
%\end{cases} $$
%and that also 
%$$|\mbox{TSS($G[U_{i}-L_{i}]$, $k_{i}$)}|=\begin{cases} 
%|\mbox{TSS($G[U_{i+1}-L_{i+1}]$, $k_{i+1}$)}| +1 & \mbox{if $k_{i}(v_{i}) > \d_{i}(v_{i})$,} \\
%|\mbox{TSS($G[U_{i+1}-L_{i+1}]$, $k_{i+1}$)}| & \mbox{otherwise.} 
%\end{cases} $$
Hence, by (\ref{eq-opt2}) we have
$$|\Sol{MTS}{G[U_{i}]}{k_{i}}| \leq |\Sol{TSS}{G[U_{i}-L_{i}]} {k_{i}}|.$$
The optimality of TSS implies 
$|\Sol{MTS}{G[U_{i}]}{k_{i}}| = |\Sol{TSS}{G[U_{i}-L_{i}]} {k_{i}}|.$ \qed 

\remove{
Then in round $\ell$ the algorithm MTS($G$, $t $) operates on $F(\ell-1)$ as the algorithm  TTS  operates on $F(i)=F(\ell-1)$ but with different thresholds (i.e., $k_{\ell-1}()$ in place of $k_{i}(v)$).

By lemma  \ref{lemmaOptimal} we have $|TSS(F(i),k_{\ell-1} )|\leq |TSS(F(i),k_{i} )|$. Hence there exists an optimal solution for 
$(F(i),k_{i}())$ that can be obtained by operating as the algorithm TTS($G$) operates on $F(i)$ with thresholds $k_{\ell-1}()$.
In particular $\exists S^{*}$ for $F(i)$ such that if TTS($G$) operating on $F(i)$ with thresholds $k_{\ell-1}()$ selects $v_\ell$ then $v_\ell \in S^{*}$ and viceversa.

After that the node $v_\ell \notin U_\ell - L_\ell$ and by inductive hypothesis we know that $|\mbox{MTS($F(\ell)$, $t_\ell()$)}|=|\mbox{TSS($F_\ell$, $t_\ell()$)}|$
\qed 
}

%\medskip
%By Theorem \ref{optimal} and the results in \cite{CGM+}, we get the following.
%
%
%
%\begin{corollary}
%The algorithm  MTS($G$, $t$) returns an optimal solution   whenever the input graph is either a path, a cycle, or a clique.
%\end{corollary}

\subsection{Optimality on Dense graphs}
We prove the optimality of algorithm $MTS$  for a class of dense graphs known as Ore graphs, whenever the threshold of each node is equal to 2. \\ 
An {\em Ore graph}  $G=(V,E)$ has the property  that for  $u,v\in V$ 
$$\mbox{if $(v,u) \not\in E$  then 
  $|\Gamma(u)| + |\Gamma(v)|  \geq n.$}$$
  It was proved in \cite{Fr+}  that  any Ore graph $G$ 
with $t(u)=2$ for each $u \in V$,  admits an optimal target set of size two.
We will prove that algorithm $MTS$ works optimally on $G$.
\begin{theorem} \label{Ore}
The algorithm MTS($G$, $t$) outputs an optimal solution whenever 
 $G=(V,E)$ is an  Ore graph and  $t(u)=2$, for each $u \in V$.
\end{theorem}
\proof
First we prove some claims that will be useful in the sequel.\\
Since each node in $G$ has threshold equal to 2 then the algorithm MTS selects and adds to the solution $S$ at least two nodes.
Let  $u_1$ and $u_2$ be the first and the second node that the algorithm MTS selects and adds to   $S$ and let $\tau_1$ and $\tau_2$ be the iterations in which such nodes  are selected, respectively.
Furthermore, let $X \subseteq V$ be such that $|\Gamma(x)|< n/2$  for each $x \in X$ and
let $Y \subseteq V$ be such that $|\Gamma(y)| \geq n/2$  for each $y \in Y$ 
(i.e., $X$ and $Y$ are a partition for $V$).
\begin{itemize}
\item[(a1)] {\em All the nodes in $X$ form a clique.}\\
Any pair of nodes in $X$ are neighbors since otherwise 
the sum of their degrees should be at least $n$ (by the definition of the Ore graph $G$)
and this is not possible by the definition of set $X$.

\item[(a2)] {\em $|X| < |Y|$.}\\
Assume that $|X| \geq |Y|$. Since $|Y| = n -|X|$ we have $|X| \geq n/2$. By (a1) the degree of each node in $X$ is at least $n/2-1$. Moreover since each  Ore graph is connected, there is at least an edge between the sets $X$ and $Y$. Hence there is a node $x\in X$ having degree $n/2-1+1=n/2$ which contradicts the definition of $X$.
\item[(a3)] {\em  MTS($G$, $t$) first selects  all the nodes in $X$ and then the ones in $Y$.}\\
If $|X|=0$ the claim is obvious. Assume now that $|X| \geq 1$.
Since all the nodes  of $G$ have threshold equal to 2, the argmax condition of Case 3 assures that the first selected node is a node in $X$. Case 3 also occurs for each other node selected before the first node $u_1$ added to the target set (recall that in each iteration before $\tau_1$  the nodes in the residual graphs have threshold equal to 2); by (a1) such nodes are in $X$ as long as there are nodes in $X$ in the residual graphs. 
\remove{Furthermore, when node $u_1$ is selected  (Case 2 occurs in ${\tau_1}$), $\d_{\tau_1}(u_1)=1$. Hence exactly one neighbor of $u_1$ is in $Y$ (note that the only neighbor of $u_1$ at the iteration ${\tau_1}$ cannot be in $X$ since $Y \neq \emptyset$ and $G$ is a connected graph) and the proof of (a3) is completed.}
\item[(a4)] {\em If $u_1$ and $u_2$ are not neighbors then \\
(A) $u_1$ and $u_2$ have $b \geq 2$ common neighbors in $G$, \\
(B) $u_1 \in Y$ and $u_2 \in Y$.}\\
Let $|\Gamma(u_1)|=b+a_1$ and  $|\Gamma(u_2)| =b+a_2$ where
 $b$ is the number of the common  neighbors of $u_1$ and $u_2$.
 Since $u_1$, $u_2$ and their neighbors are nodes of $G$ 
(that is $n \geq 2+a_1+a_2+b$), and since  
 $u_1$ and $u_2$ are not neighbors, by the definition of Ore graph we have
$$(a_1+b) + (a_2+b) \geq n \geq 2+a_1+a_2+b,$$
that leads to have $b \geq 2$ proving (A).\\
We now prove (B).
Since $u_1$ and $u_2$ are not neighbors and  $G$ is an Ore graph we have
$|\Gamma(u_1)| + |\Gamma(u_2)| \geq n$.
Hence, by the definition of set $Y$,  at least one between $u_1$ and $u_2$ is a node in $Y$. By (a3)  the claim is proved if $u_1 \in Y$.
By contradiction, assume that $u_1 \in X$. 
In this case only nodes in $X$ are selected by the algorithm MTS in iterations up to ${\tau_1}$ (recall (a3)). Furthermore, when node $u_1$ is selected  (Case 2 occurs in ${\tau_1}$), $\d_{\tau_1}(u_1)=1$. Hence exactly one neighbor $w$ of $u_1$ is in $Y$ (note that at the iteration ${\tau_1}$ node $w$ cannot be in $X$ since $Y \neq \emptyset$ and $G$ is a connected graph). By (a1) this implies that $|\Gamma(u_1)|=|X|$. 
Furthermore, since each $x \in X$ has been selected before $u_1$ we have 
$|\Gamma(x)| \leq |X|$. Hence, 
\begin{equation}\label{X}
\mbox{each $x \in X$ has at most one neighbor in $Y$.}
\end{equation}
On the other hand, since $|\Gamma(y)| \geq n-|X|$ (recall that 
$|\Gamma(u_1)| + |\Gamma(y)| \geq n$, for each $y \in Y -\{w\}$, 
and $|\Gamma(u_1)|=|X|$) and $y$ can have at most $|Y|-1=n-|X|-1$ neighbors in $Y$,
 we have that 
\begin{equation}\label{Y}
\mbox{each $y \in Y$ has at least one neighbor in $X$,}
\end{equation}
By (\ref{X}) and (\ref{Y}) we have $|X| \geq |Y|$ which contradicts (a2).
\item[(a5)] {\em Assume that  $u_1$ and $u_2$ are not neighbors in $G$ and there exists a node $v \in Y$ such that at an iteration $\tau > \tau_2$ it holds $k_{\tau}(v)=0$. Then each node  $y \in Y \cap U_\tau$  is removed  from the residual graph when  {\em its residual threshold is 0}, that is, there exists an iteration  $\tau' > \tau$ such that $k_{\tau'}(y)=0$ (i.e. Case 1 occurs for $y$).}\\
Since $u_1$ and $u_2$ are not neighbors in $G$ and by (a4) $u_1 \in Y$ and $u_2 \in Y$, we get that both $u_1$ and $u_2$ have at least $n/2-1$ neighbors in $V-\{u_1,u_2,v\}$, while $v$ has at least $n/2-2$ neighbors in $V-\{u_1,u_2,v\}$.
Hence, there exists a set $W \subseteq V-\{u_1,u_2,v\}$ such that $|W| \geq n/2-2$ and each $w \in W$ has at least two neighbors in $\{u_1,u_2,v\}$.
By the algorithm this means that the residual threshold of $w$ is $0$ for each $w \in W$ 
(i.e.,  there is an iteration $\tau' > \tau$ such that $k_{\tau'}(w)=0$). 
Now, we note that $|\{u_1,u_2,v\} \cup W| \geq n/2+1$ and $|V-(\{u_1,u_2,v\} \cup W)| \leq n/2-1$. 
Hence, since $|\Gamma(y)| \geq n/2$ for each  $y \in Y \cap [V-(\{u_1,u_2,v\} \cup W)]$ we have that $y$ has at least $2$ neighbors in $\{u_1,u_2,v\} \cup W$, then by the algorithm its residual threshold is equal to 0.
\end{itemize}

We are ready to prove the theorem.
We will prove that $S=\{u_1, u_2\}$ at the end of algorithm MTS. 
Recall that $u_1$ and $u_2$ are selected and added to  $S$ at  iterations
$\tau_1$ and $\tau_2$, respectively.
We distinguish two cases depending whether  $u_1$ and $u_2$ are  neighbors in $G$ or not.\\
$\bullet$  \ \ Let $u_1$ and $u_2$ be neighbors in $G$. 
Since $u_1$ is the first node selected and put in $S$, algorithm MTS implies that Case 2 occurs for the first time at the iteration $\tau_1$
(i.e. Case 3 has occurred in each iteration previous $\tau_1$), 
that is $\d_{\tau_1}(u_1)=1$ and $k_{\tau_1}(u_1)=2$.
Hence, the only neighbor of $u_1$ in $U_{\tau_1} - L_{\tau_1}$ is $u_2$.
To complete the proof in this case we prove that $U_{\tau_1} - L_{\tau_1} = \{u_1, u_2\}$ (i.e., $\tau_2 = \tau_1+1$, $U_{\tau_2+1} - L_{\tau_2+1} = \emptyset$).
% and each node $v \in L_{\tau_2+1}$ is removed  from the residual 
% graph since its residual threshold is 0).
By contradiction, assume that $\{u_1, u_2\} \subset U_{\tau_1} - L_{\tau_1}$.
Since $G$ is a connected graph, there exists at least one neighbor of $u_2$ in $U_{\tau_1+1} - L_{\tau_1+1}$. Furthermore, by the algorithm 
$k_{{\tau_1}+1}(u_2)=1$ and  at some iteration $\tau$, with $\tau_1 +1 \leq \tau < \tau_2$,  the last neighbor of $u_2$ in $U_{\tau} - L_{\tau}$, 
say $v$,  is selected and Case 3 occurs for it. Recalling that $k_{\tau}(v)=2$ and that 
$\d_{\tau}(u_2)=1$, $k_{\tau}(u_2)=1$, the argmax condition of Case 3 should imply  that 
$\d_{\tau}(v)=1$ 
(since $v$ has been selected instead of $u_2$ at iteration $\tau$). This leads to a contradiction since these conditions would imply Case 2 for $v$.\\
$\bullet$ \ \ Let $u_1$ and $u_2$ be independent in $G$. 
In this case we will prove that each node $v \in V-\{u_1,u_2\}$ is removed  from the residual graph when  {\em its residual threshold is 0}, that is there exists an iteration  
$\tau > \tau_2$ such that $k_{\tau}(v)=0$ (i.e. Case 1 occurs for $v$).\\
By (a4) both $u_1 \in Y$ and $u_2 \in Y$; furthermore, they have $b \geq 2$ common neighbors.
Denote by $Y_1$ and $Y_2$  the sets of  neighbors of $u_1$ and $u_2$ in $Y$ , respectively, and $Y_3 = Y-(Y_1 \cup Y_2 \cup \{u_1,u_2\})$.

\smallskip

\noindent	
-- If $|Y_1 \cap Y_2| \geq 1$ then among the $b \geq 2$ common neighbors of $u_1$ and $u_2$ there is a node $v \in Y$. 
This means that $k_{\tau_2 +1}(v)=0$. 
By (a5) it holds that each node  $y \in Y$ is removed  from the residual graph when its residual threshold is 0, that is there is an iteration $\tau > \tau_2$ such that $k_{\tau}(y)=0$.
Now, we prove that also at least two nodes in $X$ are removed  from the residual graph since their residual threshold is 0. By (a1) this would imply that  within  an iteration 
$\tau' > \tau$  each  node in $X$ has residual threshold equal to 0.
Let $A$ be the set including nodes $u_1,u_2$ and $v$ and all the nodes that have the residual threshold equal to 0 by iteration $\tau$ (recall that $Y \subseteq A$ by (a5)).
We distinguish three cases according to the size of $X \cap A$.\\
If $|X \cap A| \geq 2$ then the claim trivially follows. \\		
If  $|X \cap A| = 1$. Let $x' \in X \cap A$. We prove that there exists $x \in X - \{x'\}$ that has a neighbor in $Y$ (recall that $x'$ is a neighbor of $x$ by (a1)); hence $x$ and $x'$ are the two nodes of $X$ we are looking for. 
By contradiction assume that each $x \in X - \{x'\}$ has no neighbors in $Y$. Hence,
$|\Gamma(x)| = |X|-1$ and $|\Gamma(x)| +  |\Gamma(y)| \geq n$ for each $y \in Y$.
Then it holds $|\Gamma(y)| \geq |Y|+1$, implying that $y$ has at least a neighbor in $X$ and thus a contradiction.\\
Finally, let  $|X \cap A| = 0$. By contradiction assume that each node $x \in X$ has at most one neighbor in $Y$. This and the fact that  $|X| < |Y|$ (by (a2)) imply that there exists $y' \in Y$  that has no neighbor in $X$. Hence, $|\Gamma(x)| +  |\Gamma(y')| \geq n$ for each $x \in X$. 
Then it holds $|\Gamma(y')| \geq |Y|+1$, implying that $y'$ has at least a neighbor in $X$ and thus a contradiction.

\smallskip

\noindent	
-- If $|Y_1 \cap Y_2| = 0$ then the $b \geq 2$ common neighbors 
of $u_1$ and $u_2$ are nodes in $X$. 
Let $X_b \subseteq X$ be the set of the $b$ common neighbors of $u_1$ and $u_2$.
By the algorithm, $k_{\tau_2+1}(x)=0$ for each $x \in X_b$. By (a1) we have that by  an iteration $\tau > \tau_2$  each  node in $x \in X$ has residual threshold  equal to 0.
Now, we prove that there exists a node $y \in Y-\{u_1,u_2\}$ that has residual threshold  equal to 0  within  an iteration $\tau' \geq \tau$. By (a5), this implies that  also each other node in $Y$ has residual threshold  equal to 0  within  the end of the algorithm. 
First notice that  $n/2 \leq |\Gamma(u_1)| \leq |X| + |Y_1|$, 
$n/2 \leq |\Gamma(u_2)| \leq |X| + |Y_2|$; hence,  
 $|Y_1| \geq n/2 - |X|$, $|Y_2| \geq n/2 - |X|$ and 
\begin{equation}
|Y_3|= n{-}|X|{-}2{-}|Y_1|{-}|Y_2| \leq n{-} |X|{-}2{-} (n/2{-} |X|) {-} (n/2{-} |X|)<|X|.\label{eqY3X}\end{equation}
Now, by contradiction suppose that each $y\in Y_1 \cup Y_2$ has no neighbor in $X$ and that 
each $z \in Y_3$ has at most one neighbor in $X$. 
Hence, $|\Gamma(y)| \leq n- |X|-2$ and $n \leq |\Gamma(y)| + |\Gamma(x)|$ for each $x \in X$. 
This implies that  $|\Gamma(x)| \geq |X|+2$; that is, each node $x \in X$ has at least a neighbor in $Y_3$.
By the absurd hypothesis we know also that each node in $Y_3$ has at most  one neighbor in $X$. Hence $|X| \leq |Y_3|$, that contradicts (\ref{eqY3X}).
\qed

\bigskip

A {\em Dirac graph} $G=(V,E)$ is a graph with minimum degree $n/2$. 
Since Dirac graphs are a subfamily  of  the more general class of Ore graph, Theorem \ref{Ore} 
also holds for Dirac graphs.
\begin{corollary}
Let $G$ be a Dirac graph. 
The algorithm MTS($G$, $t$) outputs an optimal solution whenever the threshold is identical for all nodes and it is equal to 2.
\end{corollary}

\subsection{Estimating the size of the solution for general graphs}
We show that, although the new algorithm in some rare cases can lead to worse solutions compared to the TSS algorithm in \cite{CGM+}, we are still able upper bound  the size of the target set obtained by 
MTS($G$, $t$) for any graph $G$. Our bound matches the one given in \cite{CGM+,ABW-10}. 
\begin{theorem}\label{teo-upper}
For any graph $G$, % and  $t$, 
the algorithm MTS($G$, $t$) outputs a target set 
$S$ of size
\begin{equation}\label{upper}
|S|\leq \sum_{v\in V} \min\left(1,\frac{t(v)}{d(v) +1}\right).
\end{equation}
\end{theorem}
\begin{proof}{}
Let $W(G[U_i])=\sum_{v \in (U_i-L_i)} \min\left(1,\frac{k_i(v)}{\delta_i(v)+1}\right)$. We prove %now 
by induction on $i$, with $i$ going from $\lambda$ down to $1$, that  
\begin{equation}\label{S}
|S\cap U_i|\leq W(G[U_i]).
\end{equation}
The   bound (\ref{upper})  on $S$ follows recalling that   $G[U_1]=G$ and $L_1=\emptyset.$
% and $\delta_n(v)=d_G(v)\geq t(v)=k_n(v)$, for each $v\in V$.
\\
If $i=\lambda$ then the unique node $v_\la$ in $G[U_\la]$  either has threshold $k_\la(v_\la)=0$  and $S_\la=\emptyset$ or the node has positive threshold  $k_\la(v_\la)>\d_\la(v_\la)=0$ and  $S_\la=\{v_\la\}$. Hence, we have $|S\cap \{v_\lambda\}|=
\min\left(1,\frac{k_\lambda(v_\lambda)}{\delta_\lambda(v_\lambda)+1}\right)=W(G[U_\lambda])$.\\
Assume now that (\ref{S}) holds  for  $1<i+1\leq \lambda$. Consider then  $G[U_i]$ and 
 the node $v_{i}$. 
%
%We notice that if case $1$ or case $2$ of the algorithm MTS($G$) holds, then
%$$U_{i+1}=U_i-\{v_i\} \mbox{ while } L_{i+1}=L_i  \mbox{, see lines $13$ and $20$}. $$
%On the other hand if case $3$ holds, then
%$$U_{i+1}=U_i \mbox{ while } L_{i+1}=L_i \cup \{v_i\}  \mbox{, see line $25$}. $$
%Hence in both cases
%By Fact \ref{fact2}, we have 
%\begin{equation}
%(U_i-L_i) \subseteq (U_{i+1}-L_{i+1})\cup\{v_i\}.
%\end{equation}
%Hence
We have 
$$|S\cap U_i|\leq |S\cap \{v_i\}|+|S\cap U_{i+1} |\leq |S\cap \{v_i\}|+W(G[U_{i+1}]).$$
We  show  now that $W(G[U_i])\geq W(G[U_{i+1}]) +|S\cap \{v_i\}|$.
We distinguish three cases according to the cases  in the algorithm MTS($G$, $t$).

\noindent \textbf{Case 1:} Suppose that Case 1 of the Algorithm MTS holds; i.e. $k_i(v_i)=0$. In this case we have that for each  $u\in \Gamma(v_i) \cap U_i$, 
$$k_{i+1}(u)=\max(k_i(u)-1,0) \quad \mbox{ and }\quad 
\delta_{i+1}(u)=\begin{cases} \delta_{i}(u) & \mbox{if $v_i \in L_i$;} \\ \delta_{i}(u)-1 & \mbox{otherwise.} \end{cases} $$
We have,
\begin{eqnarray*}
\lefteqn{W(G[U_i]){-}W(G[U_{i+1}]) =}  \\
&=& \sum_{v \in (U_i-L_i)} \min\left(1,\frac{k_i(v)}{\delta_i(v)+1}\right) -\sum_{v \in (U_{i+1}-L_{i+1})} \min\left(1,\frac{k_{i+1}(v)}{\delta_{i+1}(v)+1}\right).
\end{eqnarray*}
By Fact \ref{fact2}, if $v_i \in L_i$ we have $U_{i+1} {-} L_{i+1}= U_i{-}L_i$. Hence, 
\begin{eqnarray*}
\lefteqn{W(G[U_i]){-}W(G[U_{i+1}]) =}  \\
&=&\sum_{v \in (U_{i+1}-L_{i+1})} \left[
\min \left(1, \frac{k_i(v)}{\delta_i(v)+1}\right)- \min\left(1,\frac{k_{i+1}(v)}{\delta_{i+1}(v)+1}\right)\right]\\
&=& \sum_{v \in \Gamma(v_i)  \cap (U_{i+1}-L_{i+1}) \atop 0<k_i(v)\leq \delta_i(v) } \left[
\frac{k_i(v)}{\delta_i(v)+1}-\frac{k_{i}(v)-1}{\delta_{i}(v)+1}\right]
%\begin{cases}
%\sum \left[\frac{k_i(v)}{\delta_i(v)+1}-\frac{k_{i}(v)-1}{\delta_{i}(v)}\right] & \mbox{if $v_i \in L_i$;}\\
%\sum  & \mbox{otherwise,}
%\end{cases}
\end{eqnarray*}

Otherwise ($v_i \notin L_i$), we have $U_{i+1} {-} L_{i+1}= (U_i{-}L_i) - \{v_i\}$ and
\begin{eqnarray*}
\lefteqn{W(G[U_i]){-}W(G[U_{i+1}]) =}  \\
&=& \sum_{v \in (U_{i+1}-L_{i+1})} \left[ \min \left(1, \frac{k_i(v)}{\delta_i(v){+}1}\right)- \min\left(1,\frac{k_{i+1}(v)}{\delta_{i+1}(v){+}1}\right)\right]\\
& & +\min \left(1, \frac{k_i(v_i)}{\delta_i(v_i){+}1}\right)\\ 
&=& \sum_{v \in \Gamma(v_i)  \cap (U_{i+1}-L_{i+1}) \atop 0<k_i(v)\leq \delta_i(v) } \left[ \frac{k_i(v)}{\delta_i(v)+1}-\frac{k_{i}(v){-}1}{\delta_{i}(v)}\right] +\min \left(1, \frac{k_i(v_i)}{\delta_i(v_i){+}1}\right)\\ 
%\begin{cases}
%\sum \left[\frac{k_i(v)}{\delta_i(v)+1}-\frac{k_{i}(v)-1}{\delta_{i}(v)}\right] & \mbox{if $v_i \in L_i$;}\\
%\sum  & \mbox{otherwise,}
%\end{cases}
\end{eqnarray*}

%where in the last equation the sums are over all $v \in \Gamma(v_i)  \cap (U_{i+1}-L_{i+1})$ such that  $0<k_i(v)\leq \delta_i(v)$. 
In both cases we have
\begin{eqnarray*}
W(G[U_i]){-}W(G[U_{i+1}])&\geq& 
\sum\left[\frac{k_i(v)}{\delta_i(v)+1}-\frac{k_{i}(v)-1}{\delta_{i}(v)}\right]  
\geq 0= |S\cap \{v_{i}\}|,	
\end{eqnarray*}
where the summ is over all $v \in \Gamma(v_i)  \cap (U_{i+1}-L_{i+1})$ s.t. $0<k_i(v)\leq \delta_i(v)$.

%&\geq& 0= |S\cap \{v_{i}\}|.	

\noindent \textbf{Case 2:} Suppose that Case 2 of the algorithm  holds; i.e. $k_i(v_i)\geq \delta_i(v_i)+1$. In this case we know that for each $v \in U_i$, $k_i(v)>0$ and we have that for each $u\in \Gamma(v_i) \cap U_i$, $k_{i+1}(u)=k_i(u)-1 $ and $\delta_{i+1}(u)= \delta_{i}(u)-1. $ Furthermore, by Fact \ref{fact2} $U_{i+1} - L_{i+1}= (U_i-L_i) - \{v_i\}.$ 
Then, 
\begin{eqnarray*}
\lefteqn{W(G[U_i])-W(G[U_{i+1}]) =}\\
&=&
\sum_{v \in (U_i-L_i)} \min\left(1,\frac{k_i(v)}{\delta_i(v)+1}\right) -\sum_{v \in (U_{i+1}-L_{i+1})} \min\left(1,\frac{k_{i+1}(v)}{\delta_{i+1}(v)+1}\right)\\
&=&
 \sum_{v \in (U_{i+1}-L_{i+1})} \left[
\min \left(1, \frac{k_i(v)}{\delta_i(v)+1}\right)- \min\left(1,\frac{k_{i+1}(v)}{\delta_{i+1}(v)+1}\right)\right] \\
& & + \min\left(1,\frac{k_{i}(v_i)}{\delta_{i}(v_i)+1}\right) \\
&=& \sum_{\myatop{v \in \Gamma(v_i)  \cap (U_{i+1}-L_{i+1})} {k_i(v)\leq \delta_i(v)}} \left[
\frac{k_i(v)}{\delta_i(v)+1}-\frac{k_{i}(v)-1}{\delta_{i}(v)}\right] + 1 
\geq 1 =|S\cap \{v_{i}\}|.
\end{eqnarray*}

\noindent \textbf{Case 3:} Suppose that Case 3  holds; i.e. $k_i(v_i)\leq \delta_i(v_i)$. 
We know that  
\begin{itemize}
\item [i)]  $0<k_i(v)\leq \delta_i(v)$ for each $v \in U_i$,
\item [ii)]$ \frac{k_i(v)}{\delta_i(v)(\delta_i(v)+1)} \leq 
\frac{k_i(v_i)}{\delta_i(v_i)(\delta_i(v_i)+1)}$, for each $v\in (U_i-L_i)$, and
\item [iii)] $S\cap \{v_{i}\}=\emptyset$. 
\end{itemize}
For each $u\in \Gamma(v_i) \cap U_i$, it holds 
$k_{i+1}(u)=k_i(u)$ and $\delta_{i+1}(u)= \delta_{i}(u)-1 $. 
Furthermore, by Fact \ref{fact2} $U_{i+1} - L_{i+1}= U_i-L_i - \{v_i\}.$ 
Hence, we get that the difference $W(G[U_i])-W(G[U_{i+1}])$ is equal to 
\begin{eqnarray*}
&& { \sum_{v \in (U_i-L_i)} \min\left(1,\frac{k_i(v)}{\delta_i(v)+1}\right)
-\sum_{v \in (U_{i+1}-L_{i+1})} \min\left(1,\frac{k_{i+1}(v)}{\delta_{i+1}(v)+1}\right)}\\
%&& ={  \sum_{v \in (U_{i+1}-L_{i+1})} \left[
%\min \left(1, \frac{k_i(v)}{\delta_i(v)+1}\right)- 
%\min\left(1,\frac{k_{i+1}(v)}{\delta_{i+1}(v)+1}\right)\right]}\\
%&& \quad + \min\left(1,\frac{k_{i}(v_i)}{\delta_{i}(v_i)+1}\right)\\ 
&& =\frac{k_{i}(v_i)}{\delta_{i}(v_i)+1}  + 
\sum_{ \myatop{v \in \Gamma(v_i)  \cap (U_{i+1}-L_{i+1})}{k_i(v)\leq \delta_i(v)} } \left[
    \frac{k_i(v)}{\delta_i(v)+1} -  \frac{k_{i}(v)}{\delta_{i}(v)}
\right]\\
&& =
\frac{k_{i}(v_i)}{\delta_{i}(v_i)+1} - 
\sum_{ \myatop{v \in \Gamma(v_i)  \cap (U_{i+1}-L_{i+1})}{k_i(v)\leq \delta_i(v)} } 
          \frac{k_i(v)}{\delta_i(v)(\delta_i(v)+1)}
\end{eqnarray*}
From which we get
$$W(G[U_i])-W(G[U_{i+1}])\geq  
          \frac{k_i(v_i)}{\delta_i(v_i)+1}-\frac{|\Gamma(v_i)  \cap (U_{i+1}-L_{i+1})| \times k_{i}(v_i)}{\delta_{i}(v_i)(\delta_{i}(v_i)+1)}.$$
Using Facts \ref{fact1} and \ref{fact2}, we have that $\delta_i(v_i)=|\Gamma(v_i)  \cap (U_{i+1}-L_{i+1})|$ and consequently
$$
W(G[U_i])-W(G[U_{i+1}]) \geq 0 =|S\cap \{v_{i}\}|.
$$\end{proof}

\section{Directed graphs}

\subsection{DAGs}
A directed acyclic graph (DAG), is a digraph with no directed cycles.
When the underlying graph $G = (V, E)$ is a DAG, the Minumum Target Set problem can be solved in
polynomial time. Indeed the optimal target set solution  consists of  the  nodes having threshold larger  than  the  incoming degree, e.g. $S^*=\{v\in V \mbox{ such that } t(v)>d^{in}(v)\}$. Since the graph is a DAG,  there is at least one node $v$ that has no incoming edges.
If the node $v$ has threshold $0$ then clearly $v \notin S^*$ for any  optimal solution $S^*$. 
Otherwise,   $t(v)>0$ and  clearly $v \in S^*$ for any optimal solution $S^*$. In both cases $v \in \Active[S^*,1]$ and its outgoing neighbors can use $v$'s influence.
Considering the nodes according to a topological ordering, once a node is considered we already know that all its incoming neighbors have been considered and will be influenced.  As a consequence, if $t(v_i)\leq d^{in}(v_i)$ then clearly $v_i \notin S^*$ for any  optimal solution $S^*$. Otherwise, if $t(v_i)> d^{in}(v_i)$ then clearly $v \in S^*$ for any optimal solution $S^*$.
\begin{theorem} \label{dags}
The algorithm  MTS($D$,$t$) returns an optimal solution   for any DAG $D$ and threshold function $t$.
\end{theorem}
\proof
The key observation is that if the algorithm MTS is executed on a DAG $D=(V,E)$ then   Case $3$ never occurs. Indeed, since $D$ is a DAG, there is at least one node $s$ having no incoming neighbours.
A node $s$ that has no incoming neighbours is selected applying Case 1 or 2 depending whether its residual threshold  $k_i(s)=0$ or not. In both cases the node is removed from the graph and the remaining graph is still a DAG. As a consequence, the Case 3 never happens. Now since  Case 3 never occurs, the set $L$ will remain empty and consequently each time a node $v$ is selected, either by Case 1 or 2, the node is removed from $U$ and both the values $k(w)$ and $\delta(w)$ for each $w\in \Gamma^{\oout}(v)\cap U_i$ are decreased by one. Recalling that at the beginning $t(v)=k(v)$ and $\delta(v)=d^{in}(v)$,  for each $v \in V$, we have that Case 2 happens for a node $v$ if and only if at the beginning $t(v)>d^{in}(v).$  The proof is completed by observing that the target set identified by the algorithm MTS consists of the nodes selected by Case 2.

\subsection{Directed Trees.}

A directed tree is a directed graph which would be a tree if the directions on the edges were ignored, i.e. a polytree.

In the following we briefly provide a simple construction that shows how the MTS problem on directed tree can easily be reduced to an MTS problem on a forest of bidirectional trees. 
Consider an MTS problem on a directed tree $T=(V,E)$. Each time there is a directed edge $(u,v) \in E$ while $(v,u) \notin E$, we can split the tree in two components $T_1$ and $T_2$ which corresponds to the nodes reachable by $u$ (resp. $v$) using $E\setminus{(u,v)}$ (ignoring directions). In $T_2$ the threshold of $v$ is decreased by 
$1$, all the other thresholds remain unchanged. It is easy to see that $S$ is a target set for $T$ if and only if $S$ is a target set for $T_1$ and $T_2$. 
By recursively applying the above rule, we  remove from $T$ all the edges $(u,v) \in E$ such that $(v,u) \notin E$ and  end up with a forest of bidirectional trees $T_1,T_2, \ldots, T_r$.
\begin{corollary} \label{trees}
The algorithm  MTS($T$, $t$) can be used to obtain an optimal solution for any directed tree $T$.
\end{corollary}
\subsection{Directed Cycles.}
\begin{theorem}\label{teoC}
 The algorithm  MTS($C$, $t$) outputs an optimal solution   if  $C$ is  a directed cycle.
\end{theorem}
\proof 
If  the first selected node $v_1$ has threshold    0  then clearly $v_1\not \in S^*$ for any optimal 
solution $S^*$.\\
If the threshold of  $v_1$ is larger than
 its incoming degree then clearly $v_1\in S^*$ for any optimal 
solution $S^*$.
In both cases  $v_1\in \Active[S^*,1]$ and its outgoing neighbors can use $v_1$'s influence; 
that is, the algorithm correctly sets $k_{1}=\max(k_1 -1,0)$ for the outgoing neighbours of $v_1$.\\
%$|S|=\lceil\frac{|\{v\in V\ |\ t(v)=2\}|}{2}\rceil$ and that this value is optimal.
If threshold of each node $v\in V$ is $1\leq t(v)\leq d^{in}(v)$, we get that  during the first iteration of the algorithm  MTS($C$, $t$), the selected  node $v_1$ satisfies Case 3.
If there exist a node having incoming degree $1$, then the selected node will have both incoming degree and threshold equal to $1$. In this case there is always an optimal solution $S^*$ for $C$ such that $S^*\cap\{ v_1\}=\emptyset$. Indeed considering any optimal solution $\overline{S^*}$. If $v_1 \in \overline{S^*} $, then let $u$ be the parent of $v_1$; we have that $S^*=S-\{v_1\}\cup\{u\}$ is another optimal solution and $S^*\cap\{ v_1\}=\emptyset$. 
\\
Otherwise the cycle is bidirectional and  the selected node $v_1$ has  $t(v_1)=2$ if at least one of the nodes in $C$ has threshold $2$, otherwise  $t(v_1)=1$.
Moreover, it is not difficult  to see that there exists an optimal solution $S^*$ for $C$ such that $S^*\cap\{ v_1\}=\emptyset$. 
\\
In each case, the result   follows by Theorem \ref {trees}, since the remaining graph is a path (ignoring arc direction) on $U_2-L_2$.

\def\dn{D}
\def\udn{U}

\section{Experimental results}  \label {experiments}
We have experimentally evaluated our algorithm MTS on real-world data sets and found that it performs surprisingly well.
We conducted tests on several real networks of various sizes from the Stanford Large Network Data set Collection (SNAP) \cite{snap}, the Social Computing Data
Repository at Arizona State University \cite{ZL09}, and  
the Newman's Network data  \cite{N15}. The data sets we considered include both
networks for which a small target set  exists and networks needing a large target set, due to a community structure that appears to block the activation process (see Section \ref{sec:cor}). 

\medskip

\noindent{\bf Test Networks.} The main characteristics of the studied networks, namely being directed/undirected, number of nodes, number of edges, max degree, size of the largest connected component, clustering coefficient and modularity, are shown in Table \ref{net}.
\begin{table*}[th!]
\begin{center}
\resizebox{1.0\linewidth}{!} {
\begin{tabular}{|l|r|r|r|r|r|r|r|r|}
\hline
Name 											& Type 				& \# of nodes & \# of edges  & Max    & Size of  & Clust.  & Modularity \\ 
									&  				& & &  degree   &  the LCC  &  Coeff. &  \\ \hline
Amazon0302 \cite{snap}    & \dn 				& 262111   	&  1234877  		 &  420         &     262111      &  0.4198          &      0.6697      \\ \hline
BlogCatalog \cite{ZL09}   & \udn			  &   88784   &  4186390    	 &  9444          &     88784     &  0.4578          &      0.3182	\\ \hline
BlogCatalog2 \cite{ZL09}  & \udn 				&   97884   & 2043701 		  &    27849        &      97884     &  0.6857          &      0.3282       \\ \hline
BlogCatalog3 \cite{ZL09}  & \udn 			  &   10312   & 333983				 &    3992         &      10312      &  0.4756          &      0.2374      \\ \hline
BuzzNet \cite{ZL09} 			& \udn				&   101168  & 4284534 			 &    64289       &      101163     &  0.2508          &      0.3161     \\ \hline
Ca-AstroPh \cite{snap} 		& \udn				&   18772   & 198110 				 &    504   	     &      17903      &  0.6768          &      0.3072      \\ \hline
Ca-CondMath \cite{snap} 	& \udn				&   23133   & 93497 				 &    279   	     &      21363       &  0.7058          &      0.5809      \\ \hline
Ca-GrQc \cite{snap} 			& \udn				&   5242    & 14496 				 &    81          &      4158         &  0.6865          &      0.7433     \\ \hline
Ca-HepPh \cite{snap} 			& \udn				&   10008   & 118521 				 &    491         &      11204      &  0.6115          &      0.5085     \\ \hline
Ca-HepTh \cite{snap}			& \udn				&   9877    & 25998 				 &    65    	      &      8638        &  0.5994          &      0.6128     \\ \hline
Cit-HepTh \cite{snap} 		& \dn 	 		  &  27770    & 352807 			  &    64         &     24700       &  0.3120          &      0.7203      \\ \hline
Delicious		\cite{snap}		&	\dn         &  103144		& 1419519    		 &   3216		       &     536108	     &  0.0731          &      0.602     \\ \hline
Douban \cite{ZL09}				& \udn				&   154907  & 327162 				 &    287           &     154908       &  0.048           &      0.5773     \\ \hline
Facebook \cite{snap} 			& \udn				&  4039     & 88234 				 &    1045          &      4039        &  0.6055          &      0.8093     \\ \hline
Flikr \cite{ZL09}		 			& \udn				& 80513    & 5899822   	    &    5706         &      80513      &  0.1652          &      0.1652        \\ \hline
Higgs-twitter	\cite{snap}	& \dn     	  & 456626    & 14855842   	  &    51386   	       &      456290       &  0.1887          &      0.5046        \\ \hline
%Jazz	 \cite{N15}					& Undirected	&  198	    & 2742	 				 &   100  		       &     198	        &   17899	  &  0.6334          &      0.4428      \\ \hline
%Karate		\cite{N15}				&  34		    & 78		 			&	&    17  		&   5	          	    &   45			&  0.5879          &      0.4151      \\ \hline
Last.fm \cite{ZL09}	 			& \udn				& 1191812   & 5115300   		 &    5140		       &      1191805    &      0.1378        &     0.1378      \\ \hline
%Les Miserable  \cite{N15} & Undirected  & 77		 		&379			   		 &    43  		       &      77	        &   894     &  0.6614          &      0.3384     \\ \hline
Livemocha \cite{ZL09}			& \udn				& 104438    & 2196188   		 &    2980		       &      104103       &  0.0582          &      0.36     \\ \hline
%Phy												& 37149   	&174163		   	&	&    178 		&   ?       &      19979      & 1283601   &  0.8656       &       ?     \\ \hline
Power grid     \cite{N15} & \udn			  & 4941		 	&6594			   		 &    19  		      &      4941          &  0.1065          &      0.6105     \\ \hline
%Taroexchange								& 22    		&39				   	&	&    6   		&   5       &      22         &   10      &  0.3394          &      0.4474 \\ \hline
%Trainbombing								& 64   			&266			   	&	&    46   	&   4       &      64         &   593     &  0.7109          &      0.2315     \\ \hline
Youtube2 \cite{ZL09}			& \udn				& 1138499 	&2990443	   		&    28754         &      1134890     &  0.1723          &     0.6506   \\ \hline
%Zebra 											& 27    		&111			   	&	&    14 		&   4       &      23         &   312     &  0.845           &      0.2768     \\ \hline
\end{tabular}
}
\end{center}
\caption{Networks. (\dn\ = Directed, \udn\ = Undirected) 	\label{net} }
\end{table*}

\noindent{\bf The competing algorithms.}
 {Several heuristics devoted to  compute small size target sets have been proposed in the literature;
they are typically classified in: \textit{additive} algorithms \cite{CWW-10,CWY09,KKT-03} and 
\textit{subtractive} algorithms \cite{CGM+,SEP,Kundu2015107} (depending on whether they focus on the addition of nodes to the target set or removal of nodes from the network).
Additive algorithms typically follow a greedy strategy which adds iteratively a node to a set $S$ until $S$ becomes a target set. Among them we compare our algorithm to  an (enhanced) 
{\em Greedy} strategy, in which nodes of maximum degree are iteratively inserted in the set $S$
and pruned from the graph. 
Nodes that remains with zero threshold are simply eliminated from the graph, until no node remains.
Subtractive algorithms, on the other hand, continuously prune the graph, according to a specific rule. The target set is determined, in this case, by the remaining nodes or by nodes that, during the pruning stage, cannot be influenced by the remaining nodes.
Among subtractive algorithms we evaluated two algorithms: {\em TIP\_DECOMP} recently presented in
\cite{SEP}, in which nodes minimizing the difference between  degree and threshold  are pruned from the graph 
until a ``core''  set is produced; {\em TSS} \cite{CGM+} which is a preliminary version of the algorithm MTS presented in this paper.
 TSS,  TIP\_DECOMP and  Greedy represent the state of the art strategies for the Minimum Target Set problem.

 %We compare the
%performance of our algorithms toward that of the best, to our knowledge, computationally
%feasible algorithms in the literature \cite{CGM5,SEP}.
%%It is worth to mention that the following competing algorithms were initially designed  for the Maximally Influencing Set problem, where the goal is to identify a set $S \subseteq V$ such that its cost is bounded by a certain budget $\beta$ and the
%%activation process activates as much nodes as possible. In order to compare such algorithms toward our strategies, for each algorithm we performed a binary search in order to find the smallest value of $\beta$ which allow to activate all the nodes of the considered graph.
%Namely, we compare to  Algorithm {\em TIP\_DECOMP} recently presented in
%\cite{SEP}, in which nodes minimizing the difference between  degree and threshold  are pruned from the graph 
%until a ``core''  set is produced.
 %We also 
%Finally we compare to the algorithm {\em TSS} \cite{CGM+}.
%The heuristic MTS  can be considered as an extension of TSS,  from  undirected graphs  to directed graphs, with the additional feature of taking  into account the influence that a node selected in Case 3 may apply on his outgoing neighbors. {\color{blue} Indeed in the original TSS algorithm, once a node was selected in case 3, it was immediately pruned from  the graph and so its potential influence was discarded;} experimental results confirm the effectiveness of this novel implementation.  

\noindent{\bf Thresholds values.} We tested the algorithms using three categories of  threshold function:
\begin{itemize}
\item \textit{Random thresholds} where  for each node $v$ the threshold $t(v)$ is chosen uniformly at random in the interval $[1,d(v)]$;
\item \textit{Constant thresholds} where the thresholds, according to the scenario considered in \cite{SEP}, are constant among all nodes. Formally,  for each node $v$ the threshold $t(v)$ is set as $min(t, d(v))$ where $t = 2,3,\ldots, 10$  (nine configurations overall);
\item \textit{Proportional thresholds} where for each node $v$ the threshold $t(v)$ is set as $\alpha \times d(v)$ with $\alpha=0.1,0.2, \ldots, 0.9$ (nine configurations overall). Notice that for $\alpha=0.5$ we are considering a particular version of the activation process named ``majority'' \cite{FKRRS-2003}. {It is worth to
mention that when $\alpha$ is either quite close to $0$ or $1$,  the Minimum Target Set problem is much easier to solve. Indeed, for very small values of $\alpha$, a random small set of nodes is likely able to activate all the nodes, while when $\alpha$ is large, the  target set must necessarily contain almost all of the nodes in $V$. On the other hand, for intermediate value of $\alpha$, the algorithm must necessarily operate many choices and consequently the differences in performance between different algorithms are larger.}
\end{itemize}

Summarizing our experiments compare the size of the target set generated by  $4$ algorithms (MTS, TSS, TIP\_DECOMP, Greedy) on $20$ networks (see Table \ref{net}), fixing the thresholds in $19$ different ways (Randomly, Steadily with $t=2,3,\ldots,10$ and Proportionally with $\alpha=0.1,0.2, \ldots, 0.9$). Overall we performed $4 \times 20 \times 19=1,520$ tests. Since the random thresholds test settings involve some randomization, we executed each test
$10$ times. The results were compared using means of target set sizes (the observed variance was negligible).

\subsection{ Results}

\paragraph{Random Thresholds.} Table \ref{randomtest} depicts the results of the Random threshold test setting. Each number represents the average size of the target set  generated by each algorithm on each network using random thresholds (for each test, first the random thresholds have been generated and then the same thresholds values have been used for all the algorithms). The value in bracket represents the overhead percentage compared to the MTS algorithm.
Results shows that the MTS algorithm always outperforms its competitors. {The improvement depends on some structural characteristics of the network. A detailed discussion of the MTS algorithm performances on different networks will be presented in section \ref{sec:cor}.}

\begin{table*}[th!]
\begin{center}
\resizebox{0.95\linewidth}{!} {
\begin{tabular}{|l|r|r|r|r|}
\hline
%& \multicolumn{3}{c|}{TSS with Partial Incentives} &  &\multicolumn{3}{c|}{Weighted Target Set Selection}\\ \hline
Name 		& MTS &	TSS	& Greedy  &  TIP\_DECOMP\\ \hline
Amazon0302 \cite{snap}    & 14246 	& 17312 (122\%)	& 84139 (591\%) & 23657 (166\%)\\ \hline
BlogCatalog \cite{ZL09}   & 157			&	222 (141\%)		&  4507 (2871\%) & 894 (569\%)\\ \hline
BlogCatalog2 \cite{ZL09}  &	33			&	60 (182\%)		&  1842 (5582\%) & 523 (1585\%)  \\ \hline
BlogCatalog3 \cite{ZL09}  &	6				&	10 (167\%)		&    10 (167\%)  & 44 (733\%)   \\ \hline
BuzzNet \cite{ZL09} 			&	154			&	277 (180\%)		&  4742 (3079\%) & 712 (462\%)  \\ \hline
Ca-AstroPh \cite{snap} 		&	845			&	978 (116\%)		&  4555 (539\%) & 1236 (146\%)  \\ \hline
Ca-CondMath \cite{snap} 	&	1657		&	1829 (110\%)	&  5584 (337\%) & 2488 (150\%)  \\ \hline
Ca-GrQc \cite{snap} 			&	638			&	659 (103\%)		&  1408 (221\%) & 811 (127\%)  \\ \hline
Ca-HepPh \cite{snap} 			&	808 		&	878 (109\%)		&  2926 (362\%) & 1060 (131\%)  \\ \hline
Ca-HepTh \cite{snap}			&	869			&	935 (108\%)		&  2446 (281\%) & 1236 (142\%)  \\ \hline
Cit-HepTh \cite{snap} 		&	2443		&	2510 (103\%)	&  4257 (174\%) & 2960 (121\%)  \\ \hline
Delicious		\cite{snap}		&	10615		&	10882 (103\%)	& 51843 (488\%) & 38493 (363\%)  \\ \hline
Douban \cite{ZL09}				&	2405		&	2407 (100\%)	&  6868 (286\%) & 12365 (514\%)  \\ \hline
Facebook \cite{snap} 			&	165			&	189 (115\%)		&  1200 (727\%) & 169 (102\%)  \\ \hline
Flikr \cite{ZL09}		 			&	499			&	785 (157\%)		& 13104 (2626\%) & 582 (117\%)  \\ \hline
Higgs-twitter	\cite{snap}	&	935			&	1575 (168\%)	& 55532 (5938\%) & 2928 (313\%)  \\ \hline
%Jazz	 \cite{N15}					&	&	& &  \\ \hline
%Karate		\cite{N15}			&	&	& &  \\ \hline
Last.fm \cite{ZL09}	 			&	8583		&	8671 (101\%)	& 54125 (631\%) & 42852 (499\%)  \\ \hline
%Les Miserable  \cite{N15} &	&	&  & \\ \hline
Livemocha \cite{ZL09}			&	213			&	424 (199\%)		& 12568 (5900\%) & 529 (248\%)  \\ \hline
%Phy												&	&	&  & \\ \hline
Power grid     \cite{N15} &	307			&	321 (105\%)		&  1337 (436\%) & 516 (168\%)  \\ \hline
%Taroexchange							&	&	& &  \\ \hline
%Trainbombing				   		&	&	& &  \\ \hline
Youtube2 \cite{ZL09}			&	34790		& 34935 (101\%)	&142065 (408\%) & 89596 (258\%)  \\ \hline
%Zebra 										&	&	& &  \\ \hline
\end{tabular} 
}
\end{center}
\caption{Random Thresholds Results: For each network and each algorithm, the average size of the target set is depicted. 
	\label{randomtest} }
\end{table*}

\paragraph{Constant and Proportional thresholds.}

Figures \ref{fig:amazon0302}-\ref{fig:Douban2} depict the results of Constant and Proportional thresholds settings. For each network the results are reported in two separated charts: 
\begin{itemize}
	\item \textbf{Proportional thresholds} (left-side), each plot depicts the size of the target set (Y-axis), for each value of $\alpha = 0.1,0.2,\ldots, 0.9$ (X-axis) and for each algorithm (series);
\item \textbf{Constant thresholds} (right-side), each plot depicts the size of the target set (Y-axis), for each value of $t = 2,3,\ldots, 10$ (X-axis) and for each algorithm (series);
\end{itemize}

\begin{figure}[th!]
\begin{center}
		\includegraphics[width= 0.49\linewidth]{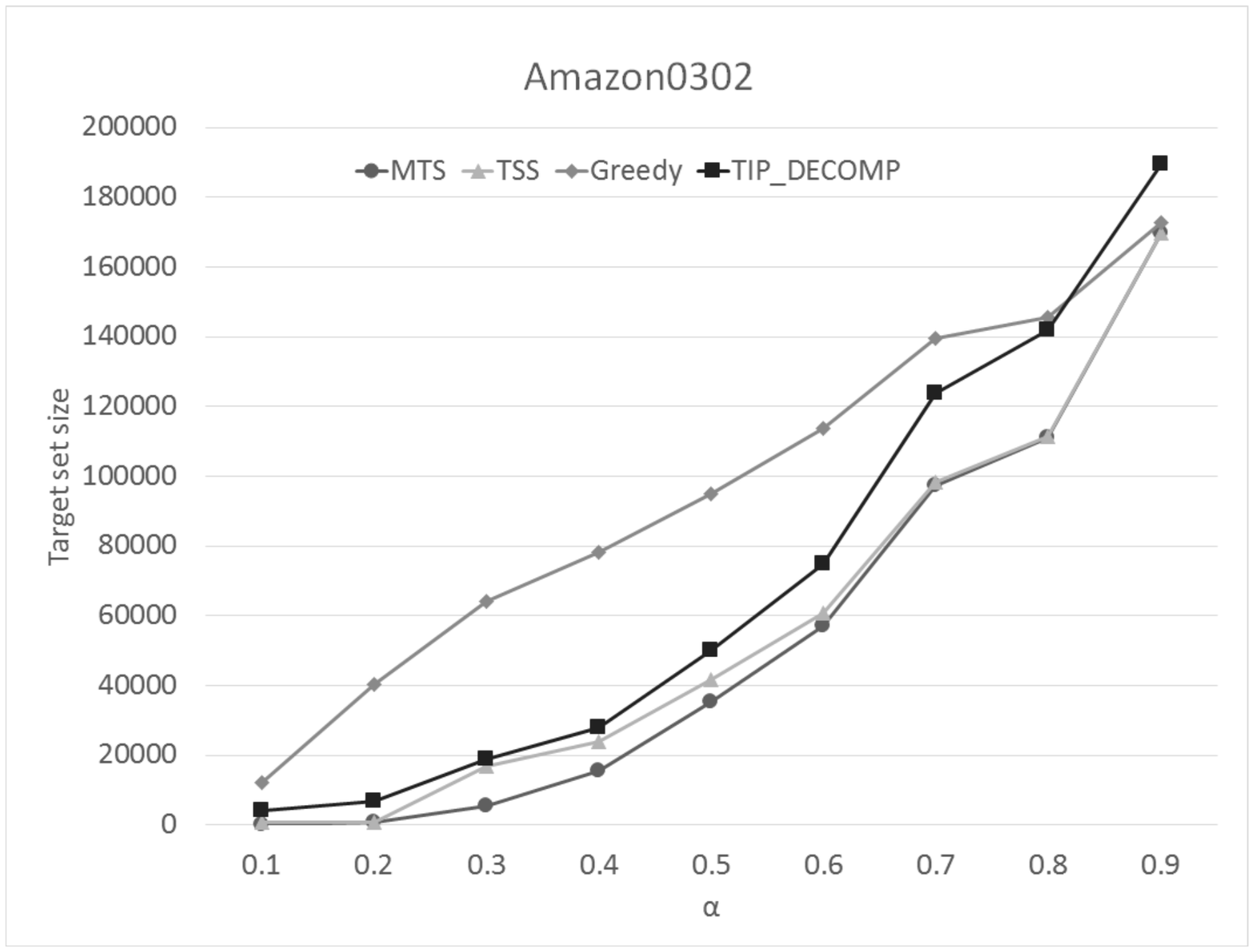}
		\includegraphics[width= 0.49\linewidth]{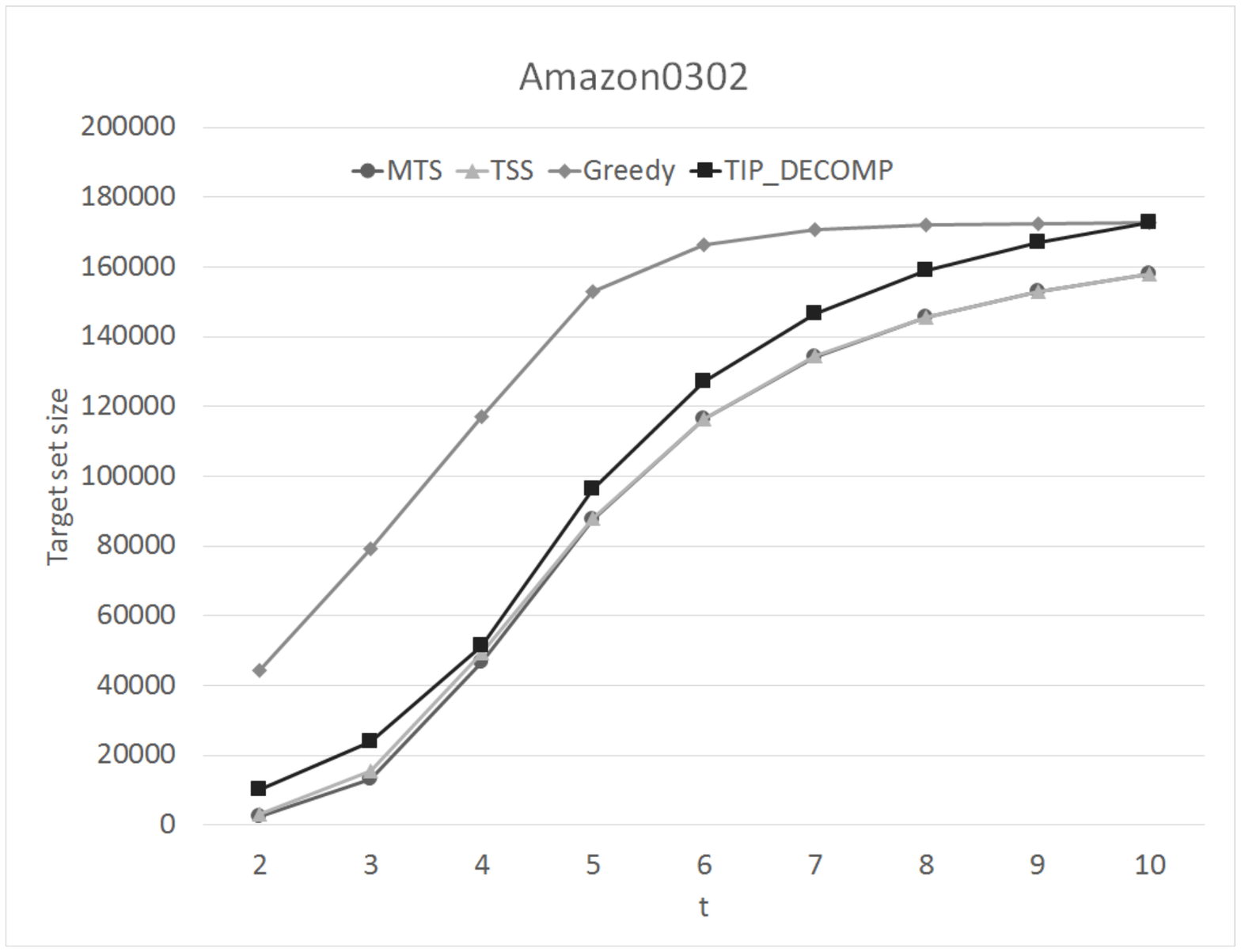}
		\includegraphics[width= 0.49\linewidth]{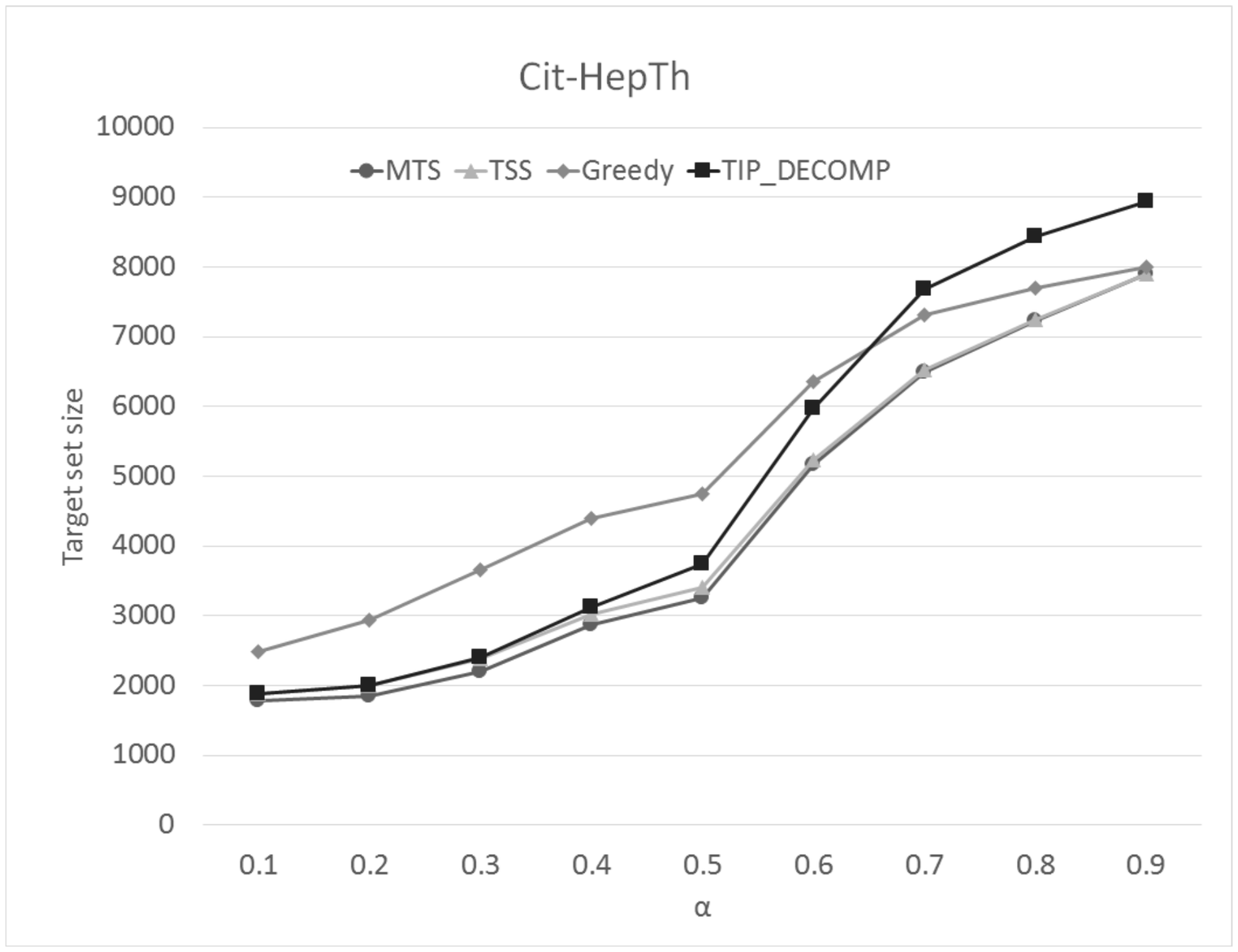}
		\includegraphics[width= 0.49\linewidth]{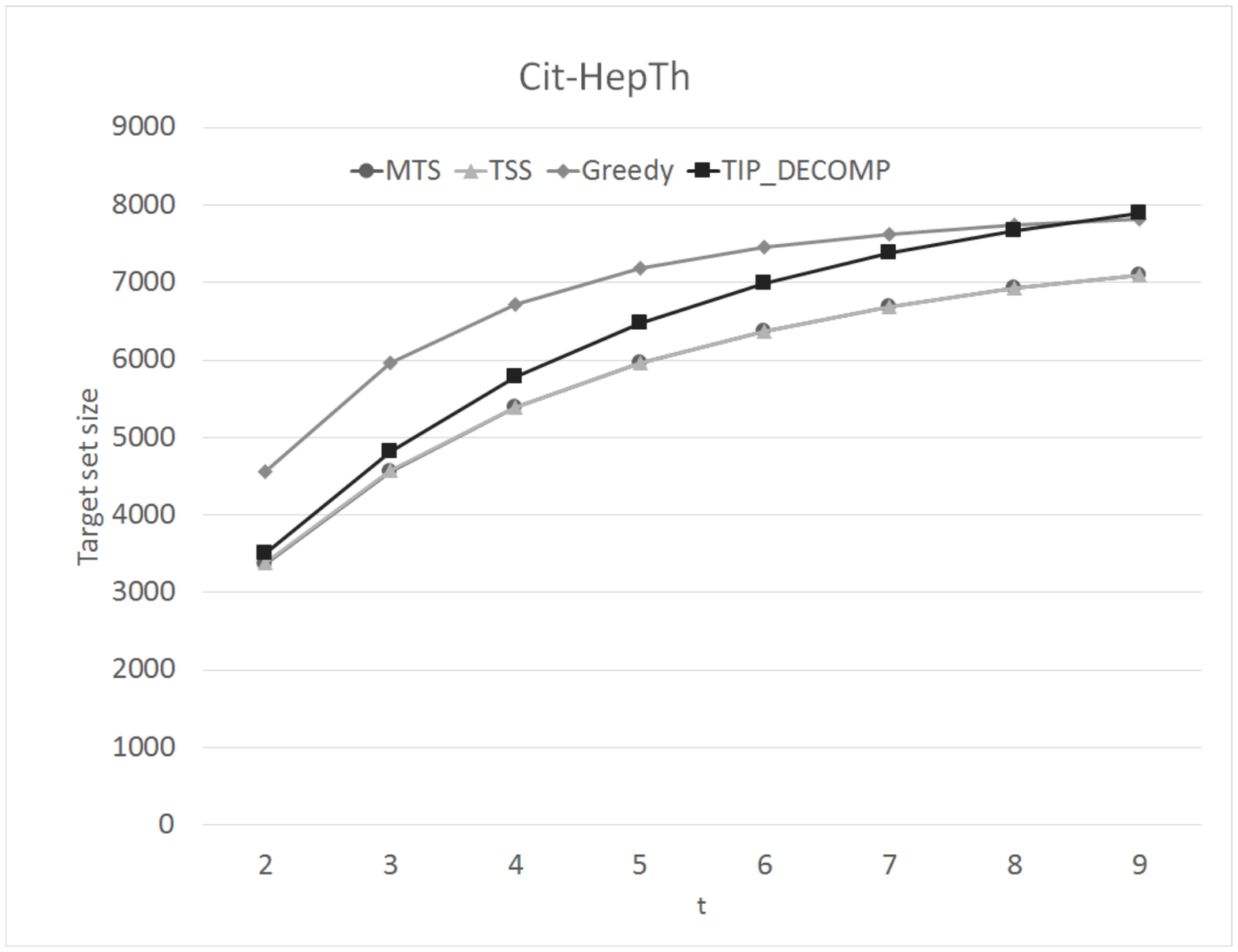}
		\caption{Constant and  Proportional Thresholds Results on Directed networks (Amazon0302, Cit-Hep-th): 
		For each network the results are reported in two separated charts: Proportional thresholds (left-side) and 
Constant thresholds (right-side).	\label{fig:amazon0302}}
	\end{center}
\end{figure}
\begin{figure}[th!]
\begin{center}
		\includegraphics[width= 0.49\linewidth]{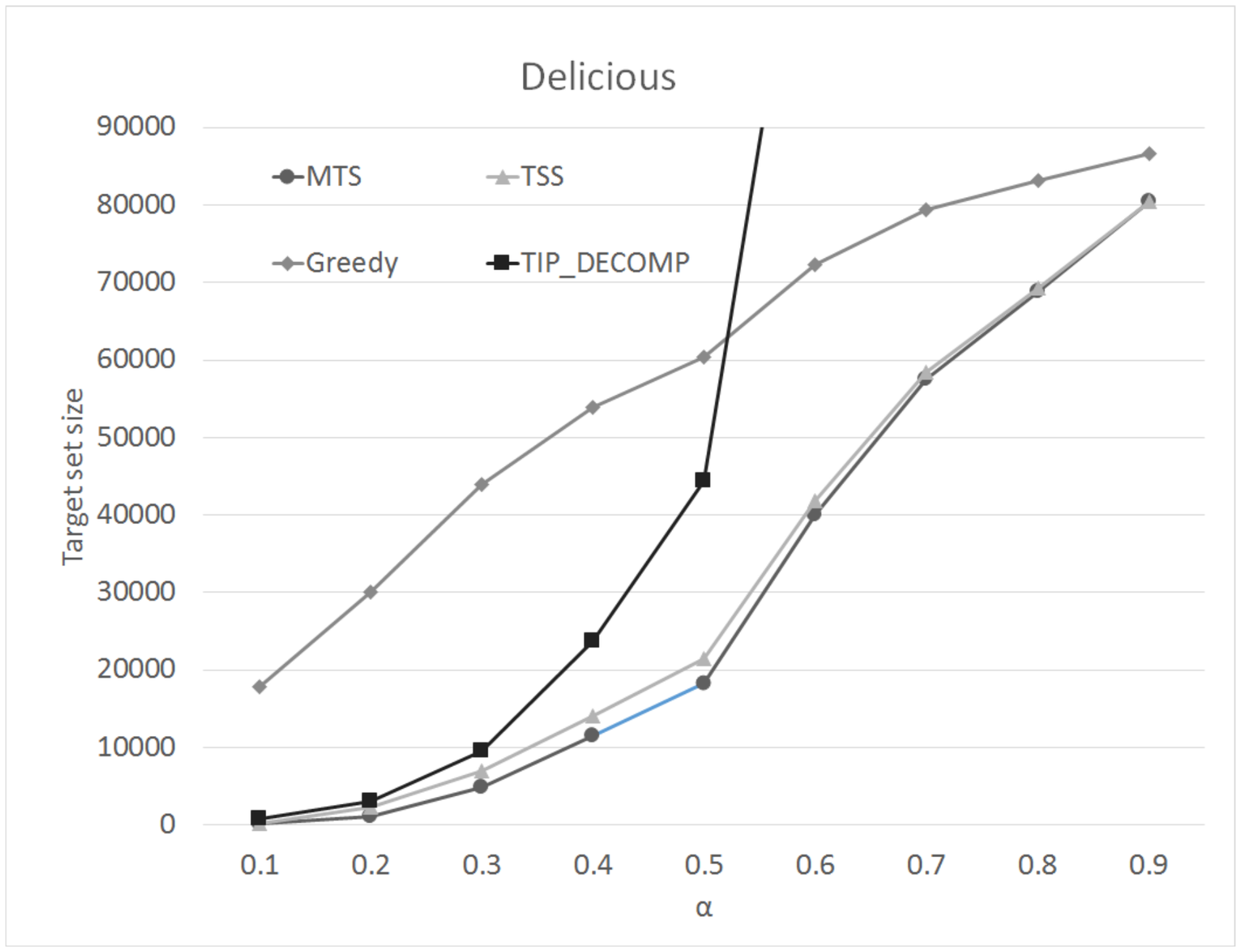}
		\includegraphics[width= 0.49\linewidth]{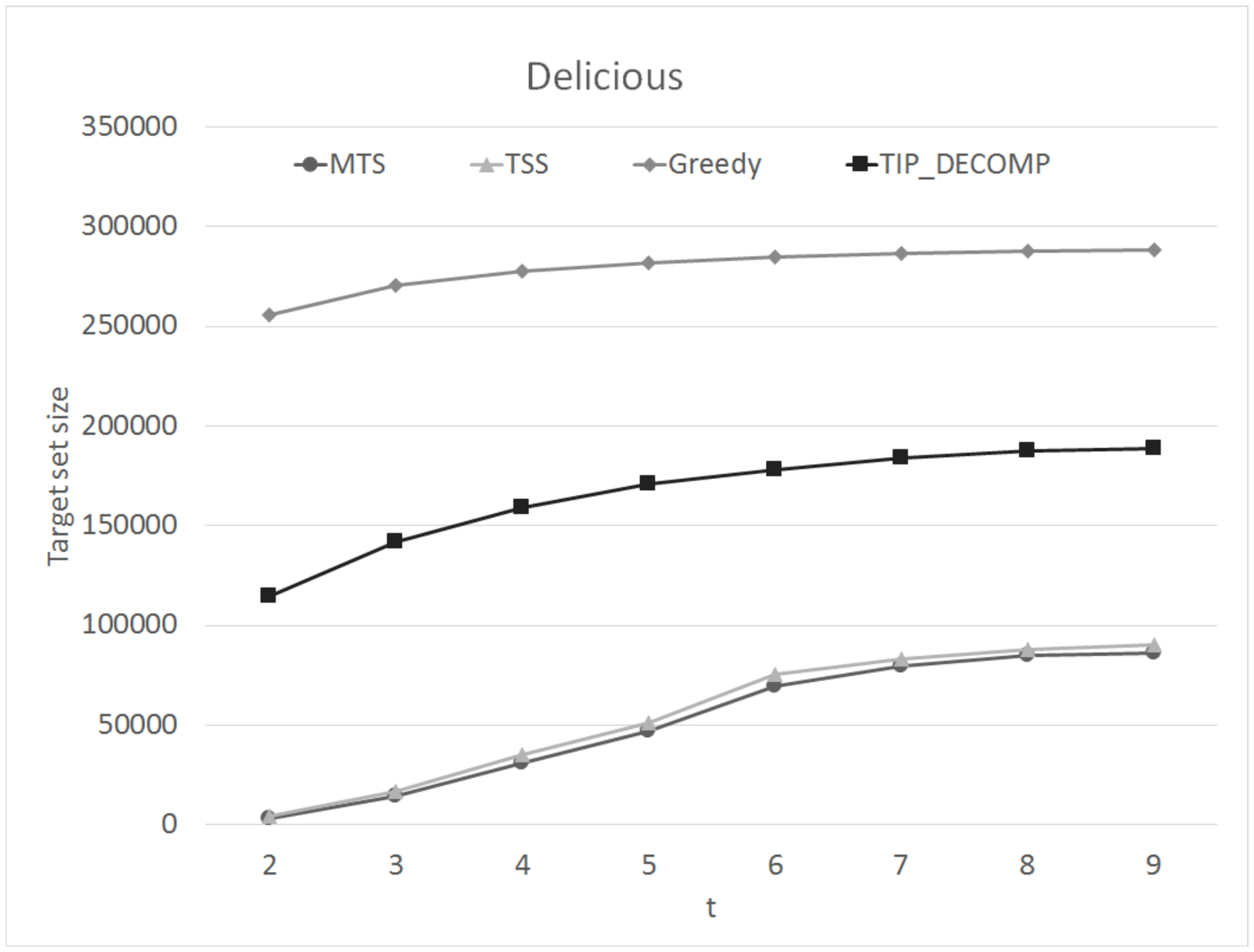}
		\includegraphics[width= 0.49\linewidth]{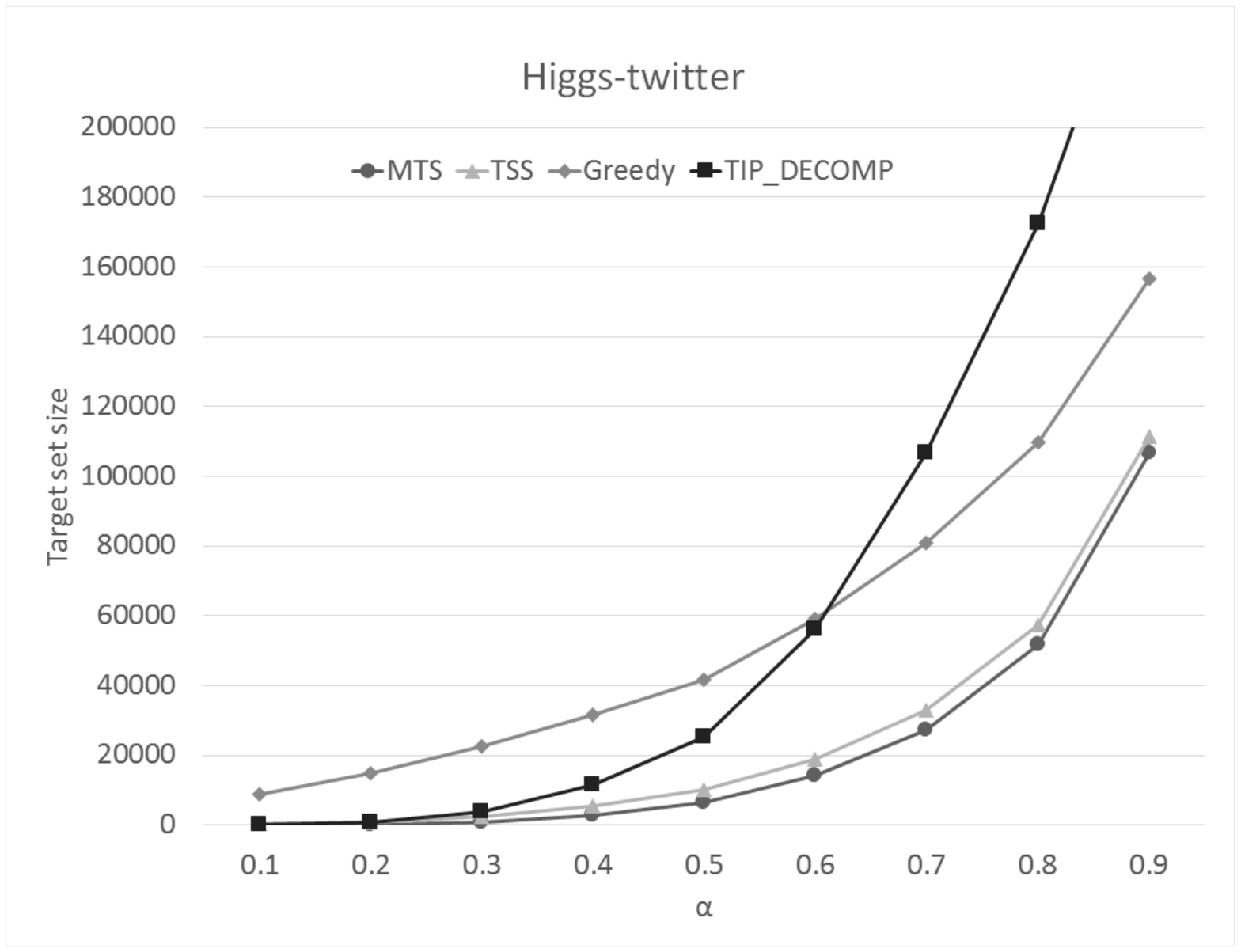}
		\includegraphics[width= 0.49\linewidth]{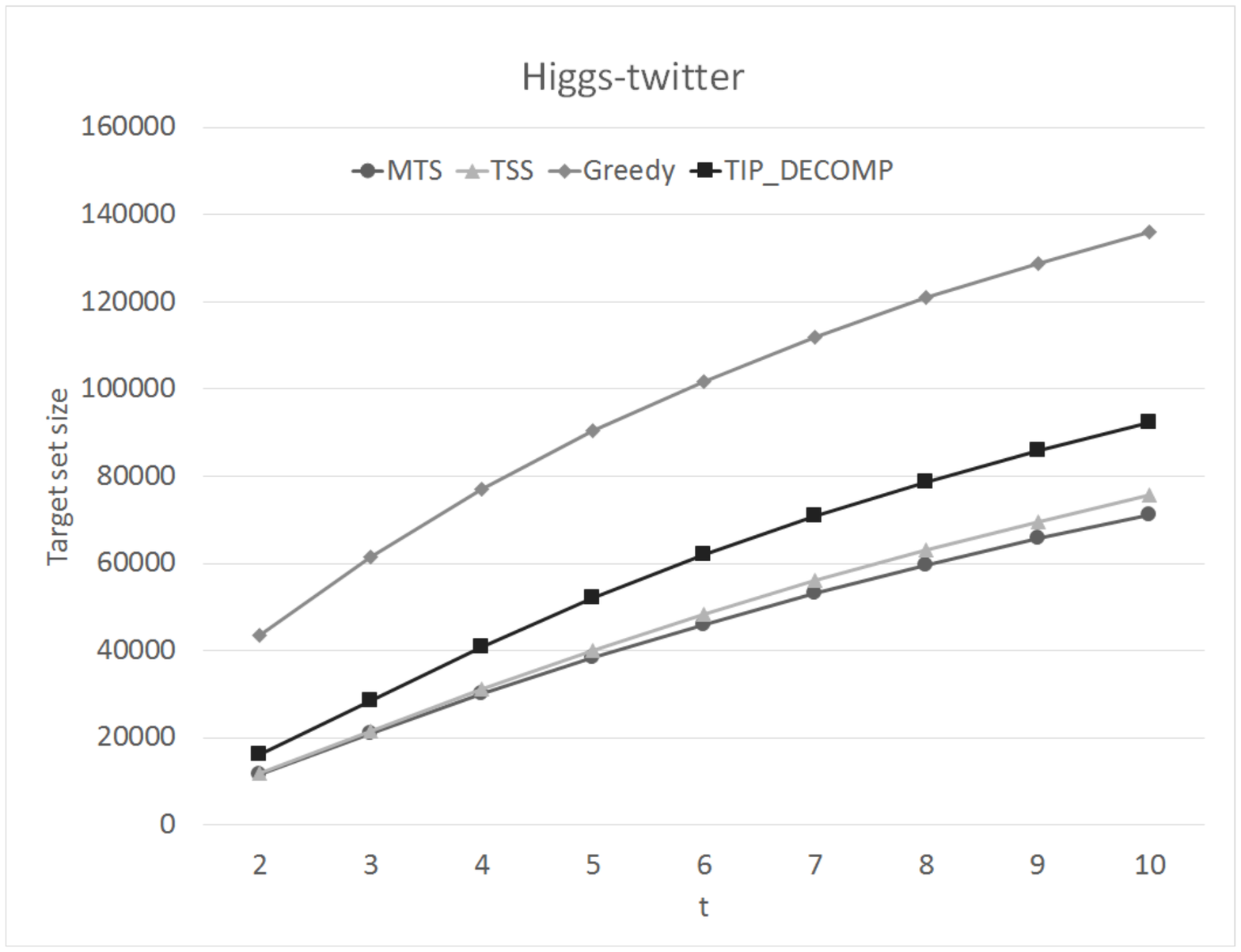}
		\caption{Constant and Proportional Thresholds Results on Directed networks  (Delicious and Higgs-twitter): For each network the results are reported in two separated charts: Proportional thresholds (left-side) and 
Constant thresholds (right-side).	\label{fig:amazon03022}}
	\end{center}
\end{figure}

\begin{figure}[th!]
\begin{center}
		\includegraphics[width= 0.49\linewidth]{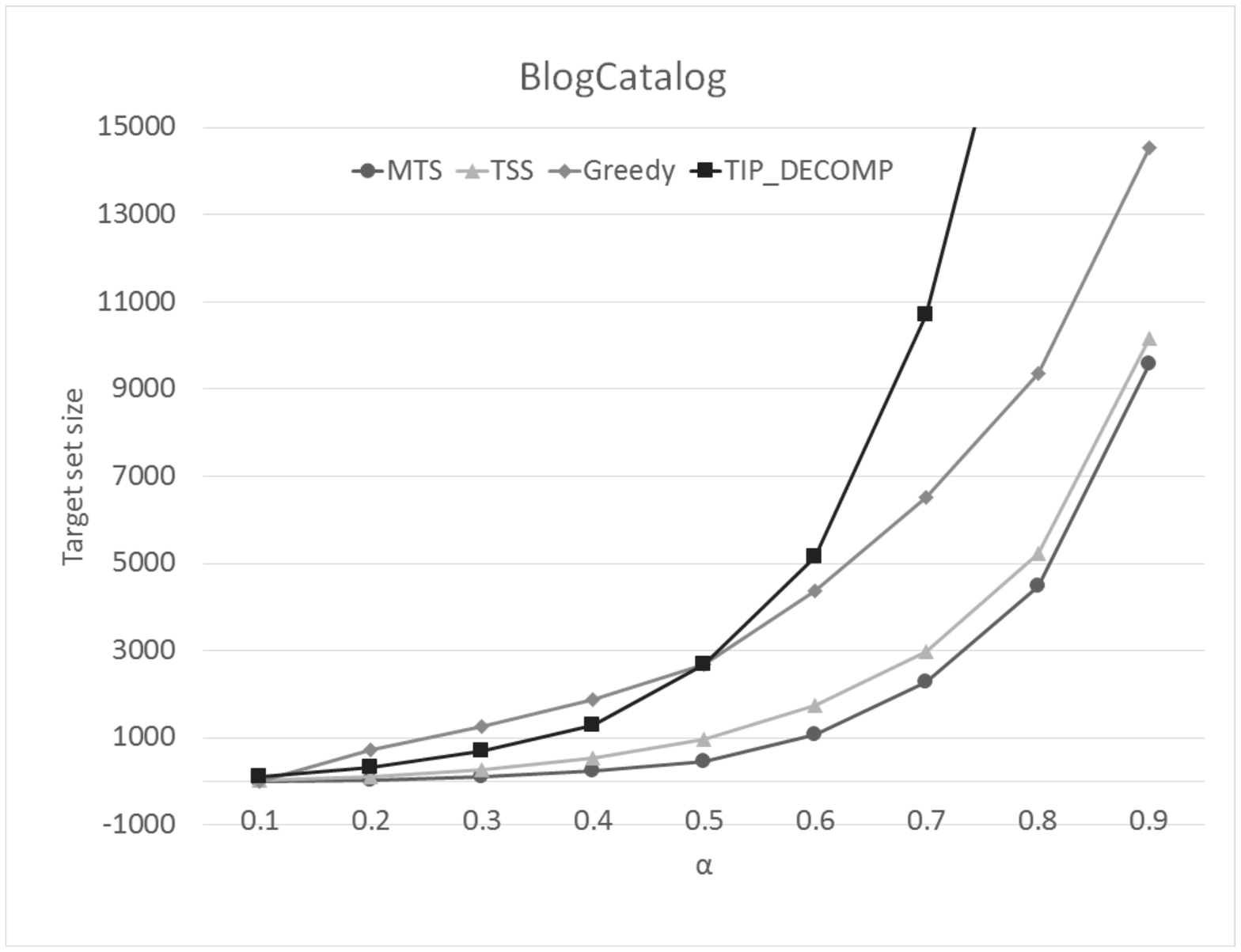}
		\includegraphics[width= 0.49\linewidth]{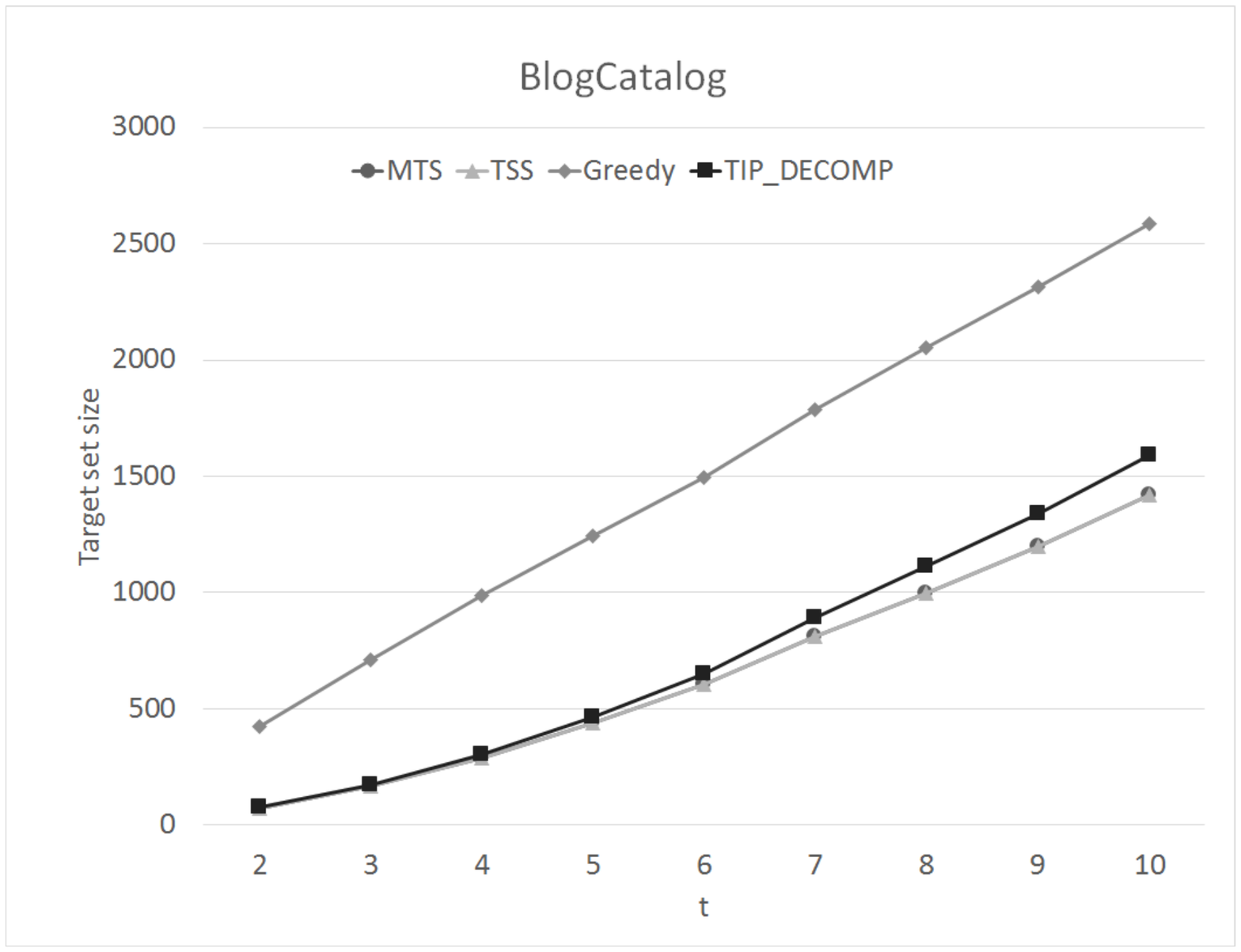}
		\includegraphics[width= 0.49\linewidth]{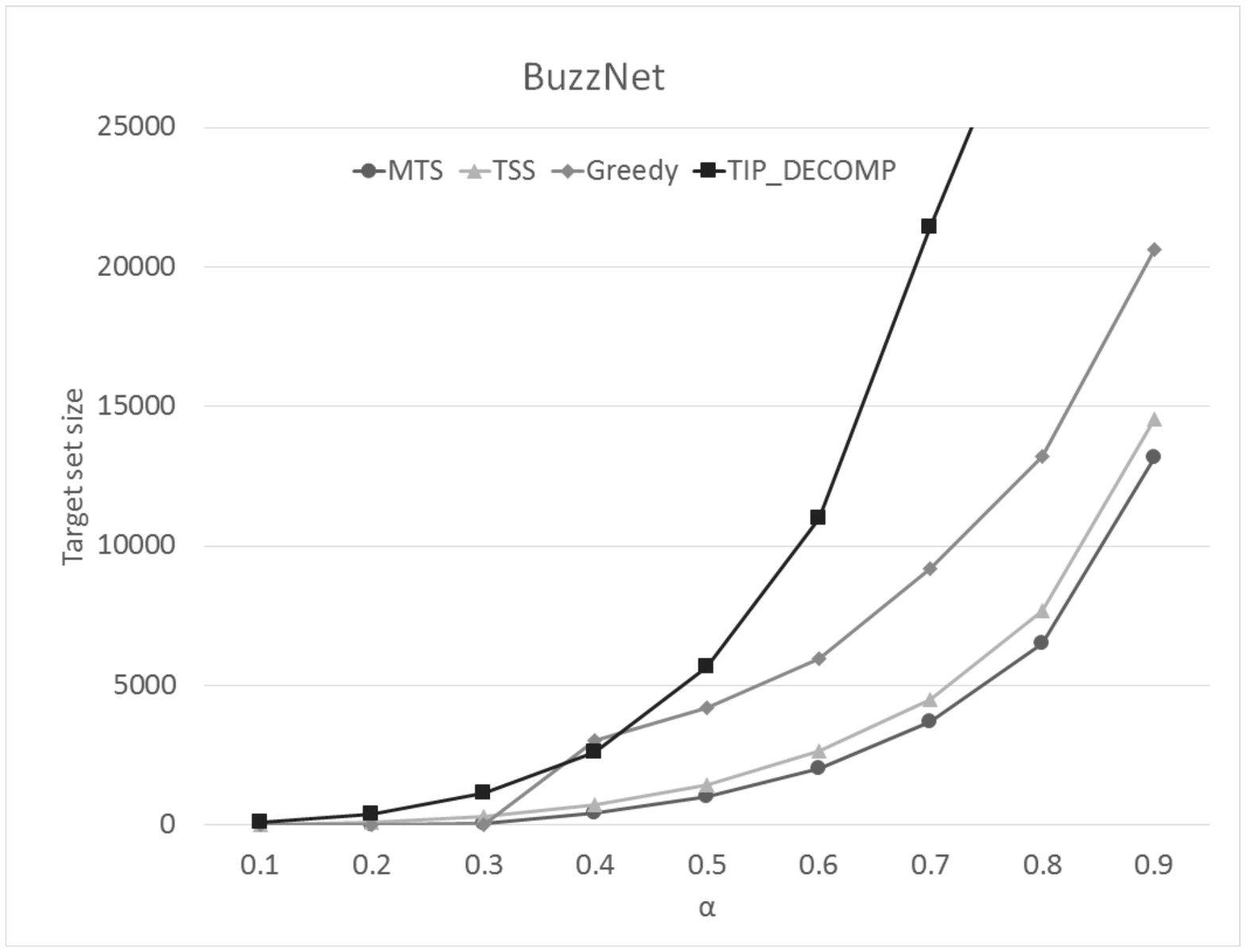}
		\includegraphics[width= 0.49\linewidth]{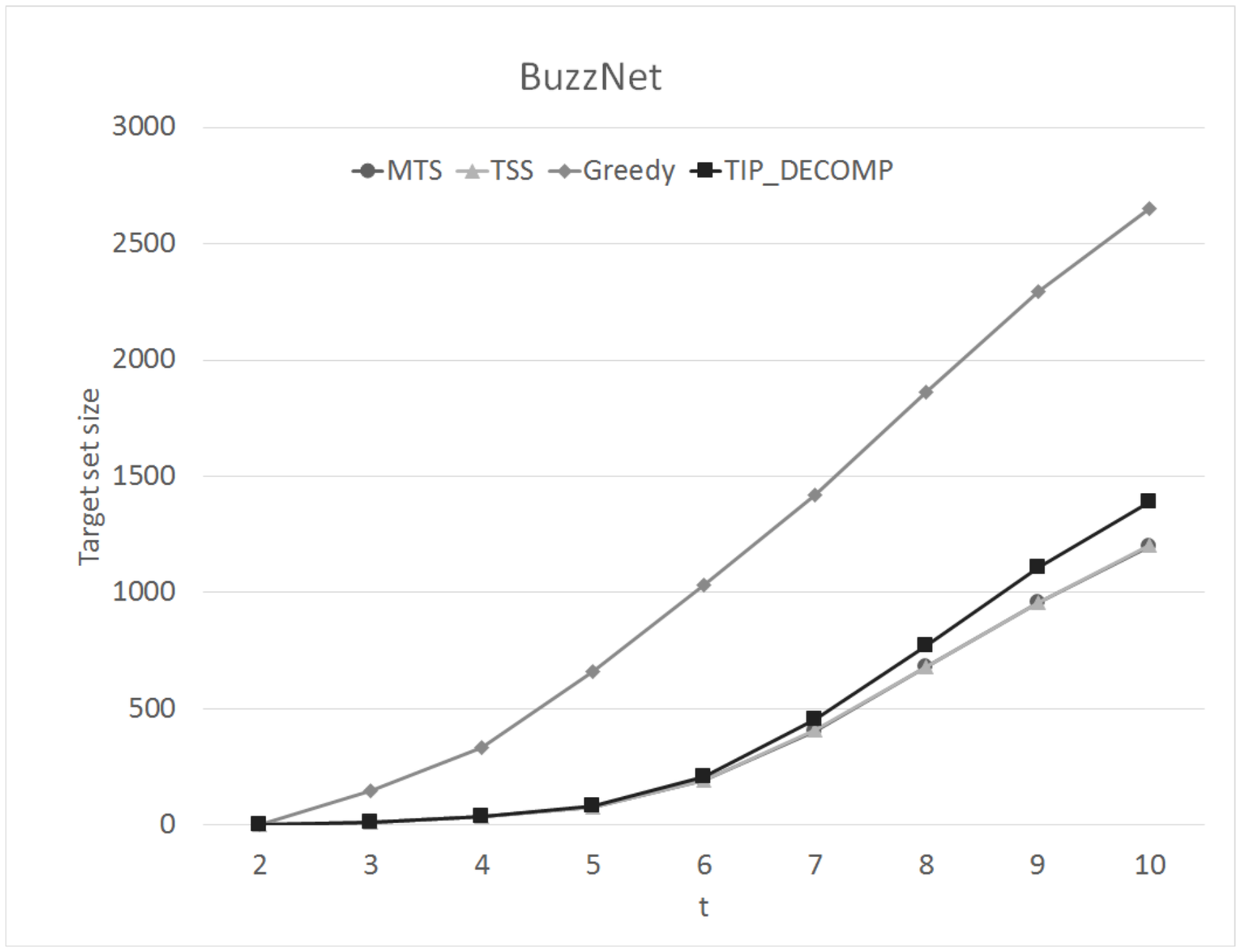}
		\includegraphics[width= 0.49\linewidth]{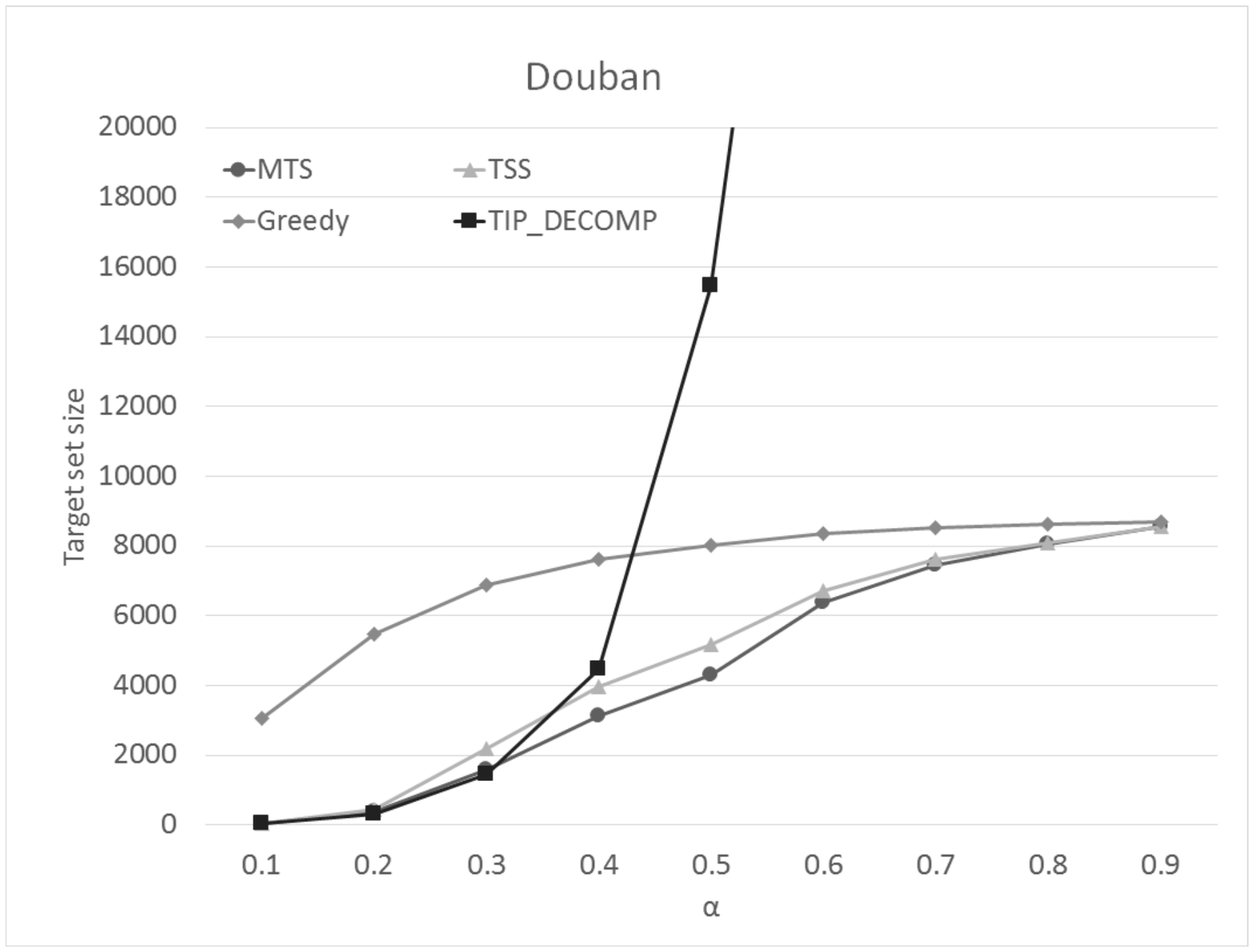}
		\includegraphics[width= 0.49\linewidth]{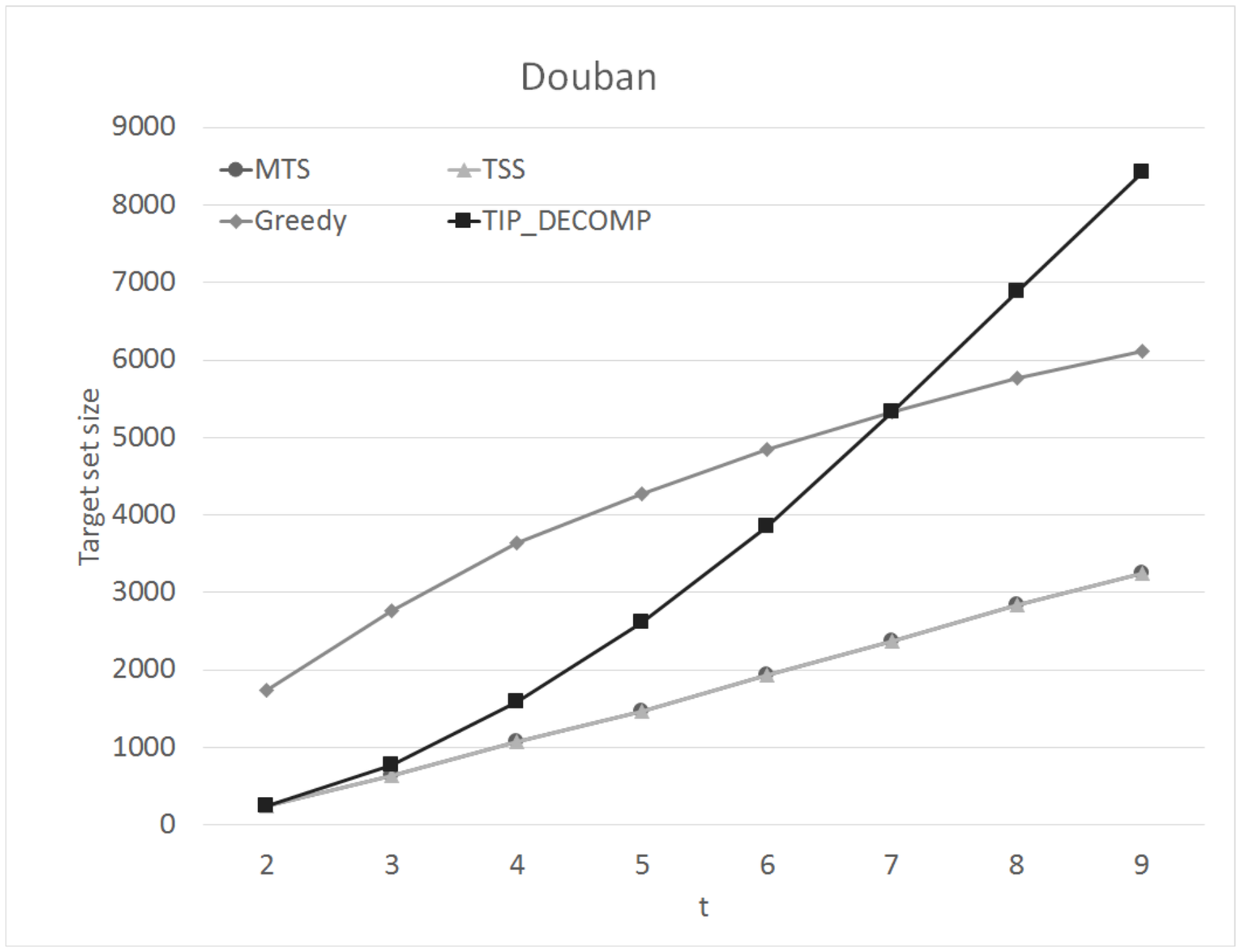}
				\caption{Constant and   Proportional Thresholds Results on undirected networks  (BlogCatalog, BuzzNet, Douban): For each network the results are reported in two separated charts: Proportional thresholds (left-side) and 
Constant thresholds (right-side).	\label{fig:Douban}}
		\end{center}
\end{figure}
\begin{figure}[th!]
\begin{center}
	\includegraphics[width= 0.49\linewidth]{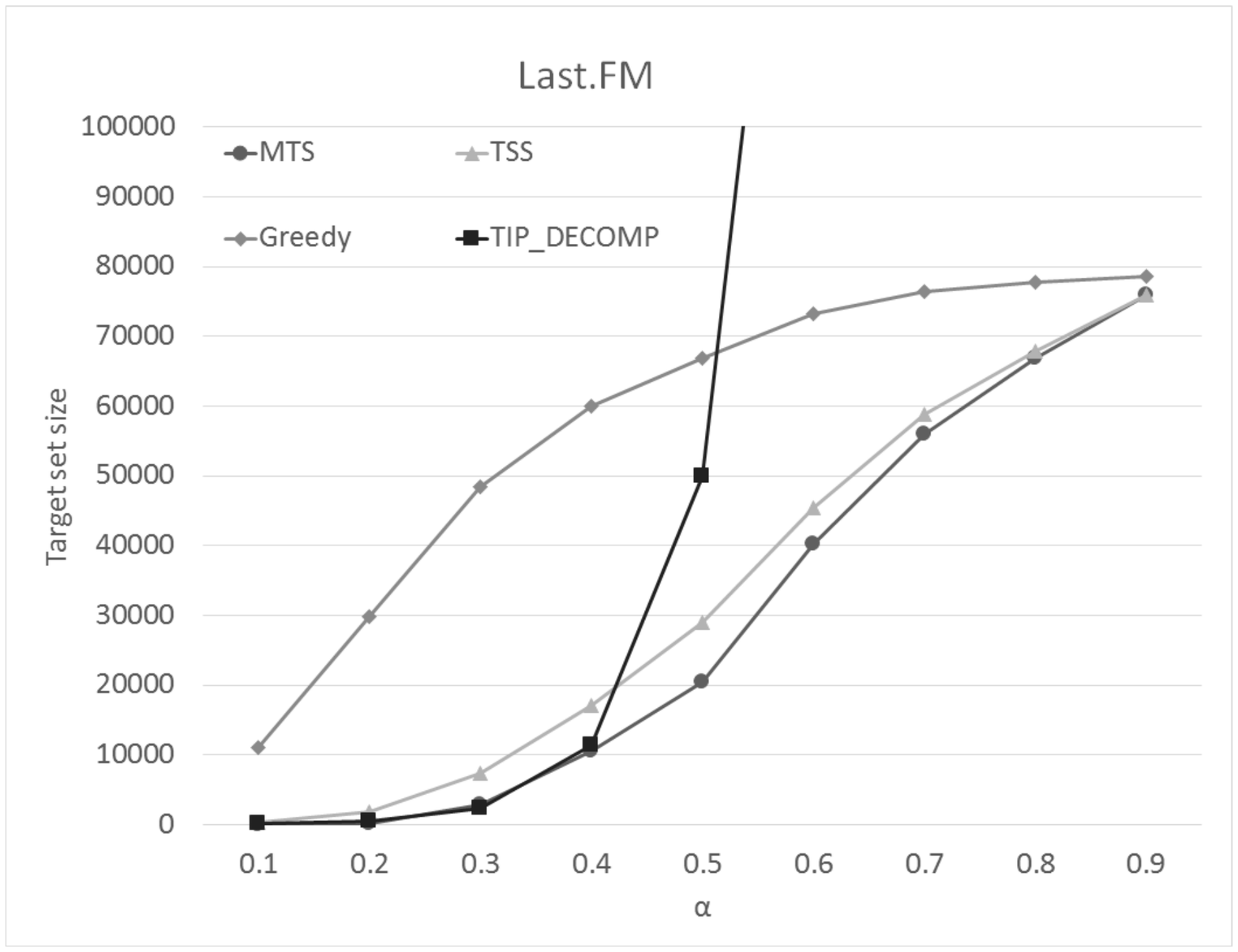}
		\includegraphics[width= 0.49\linewidth]{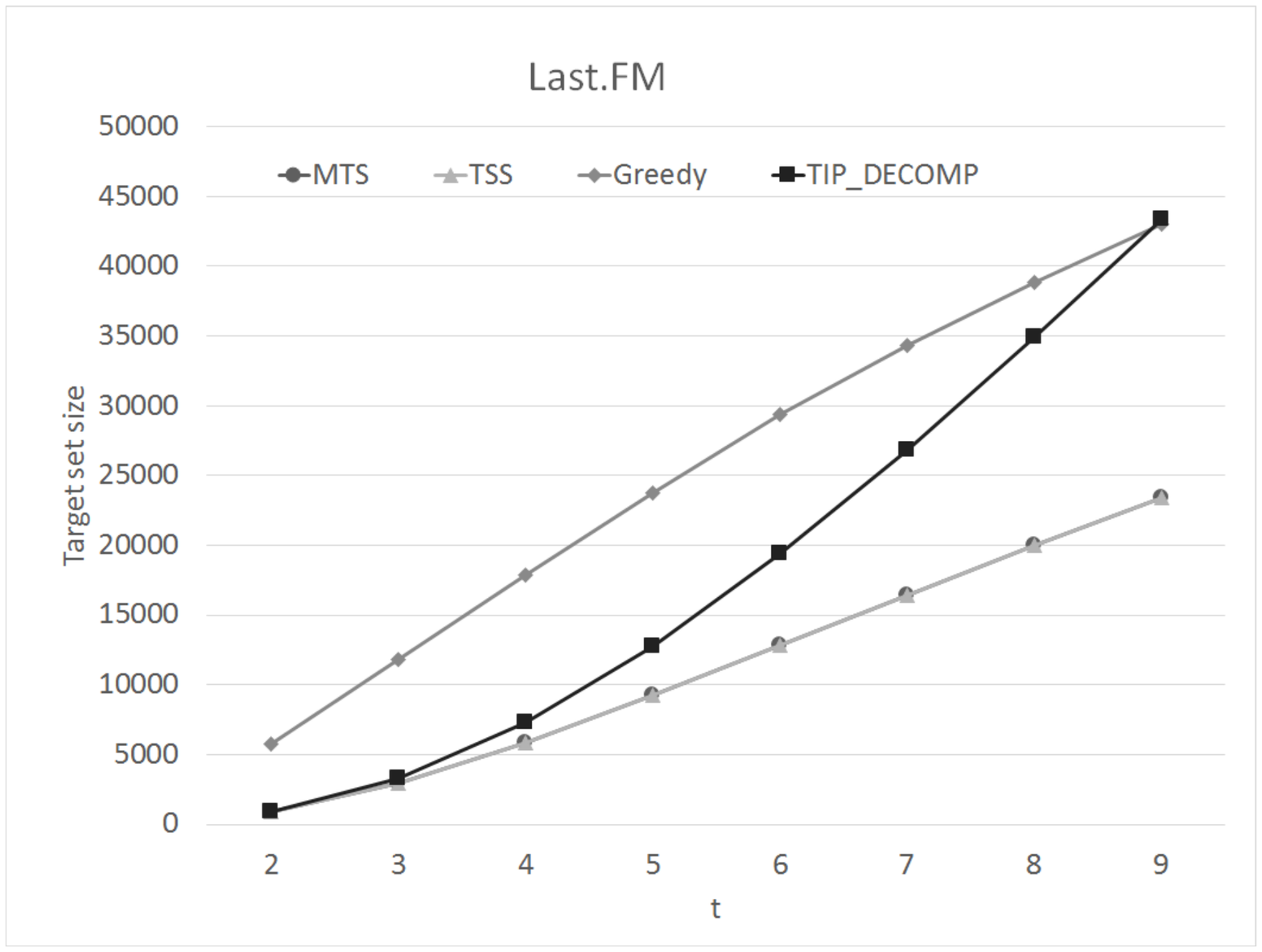}
		\includegraphics[width= 0.49\linewidth]{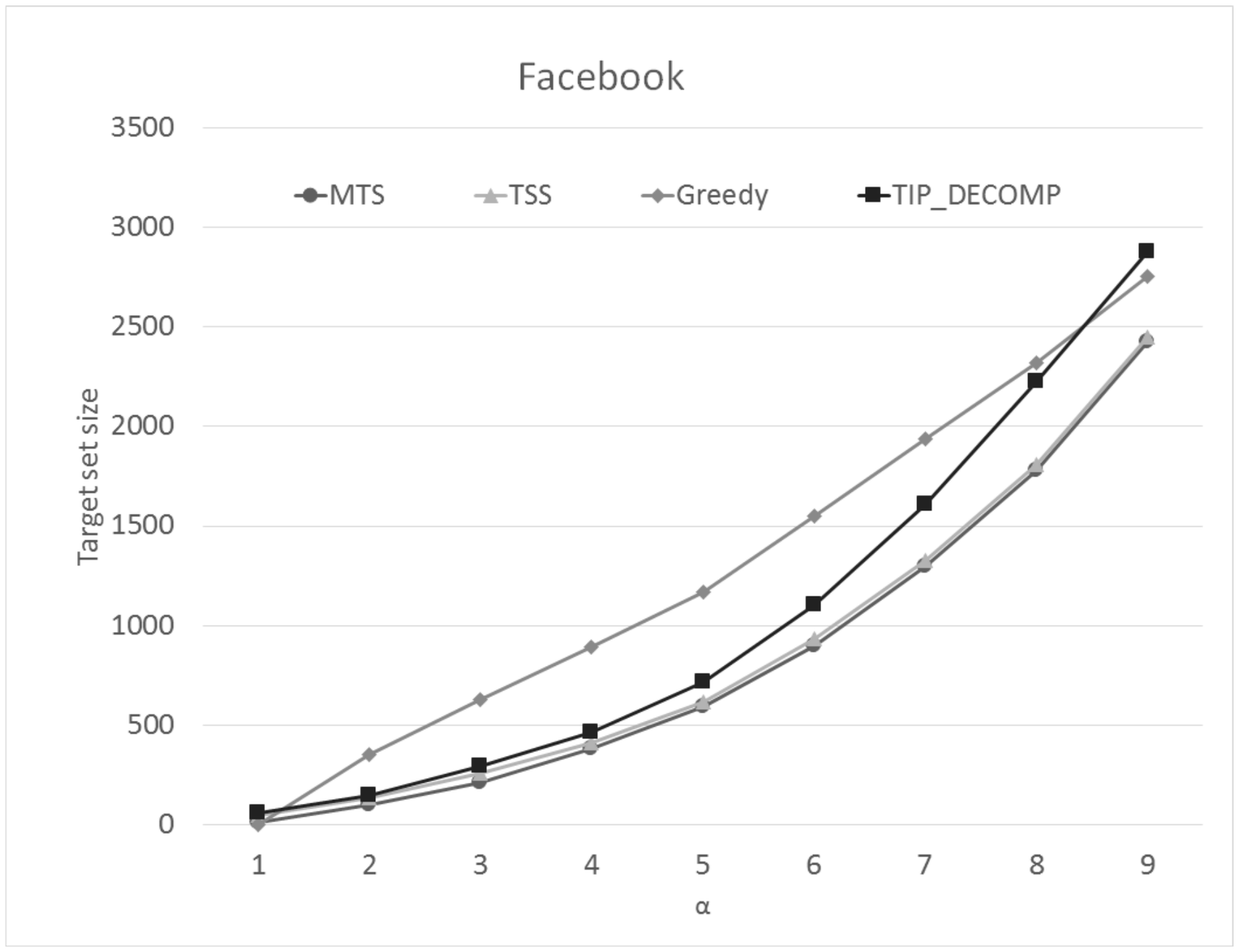}
		\includegraphics[width= 0.49\linewidth]{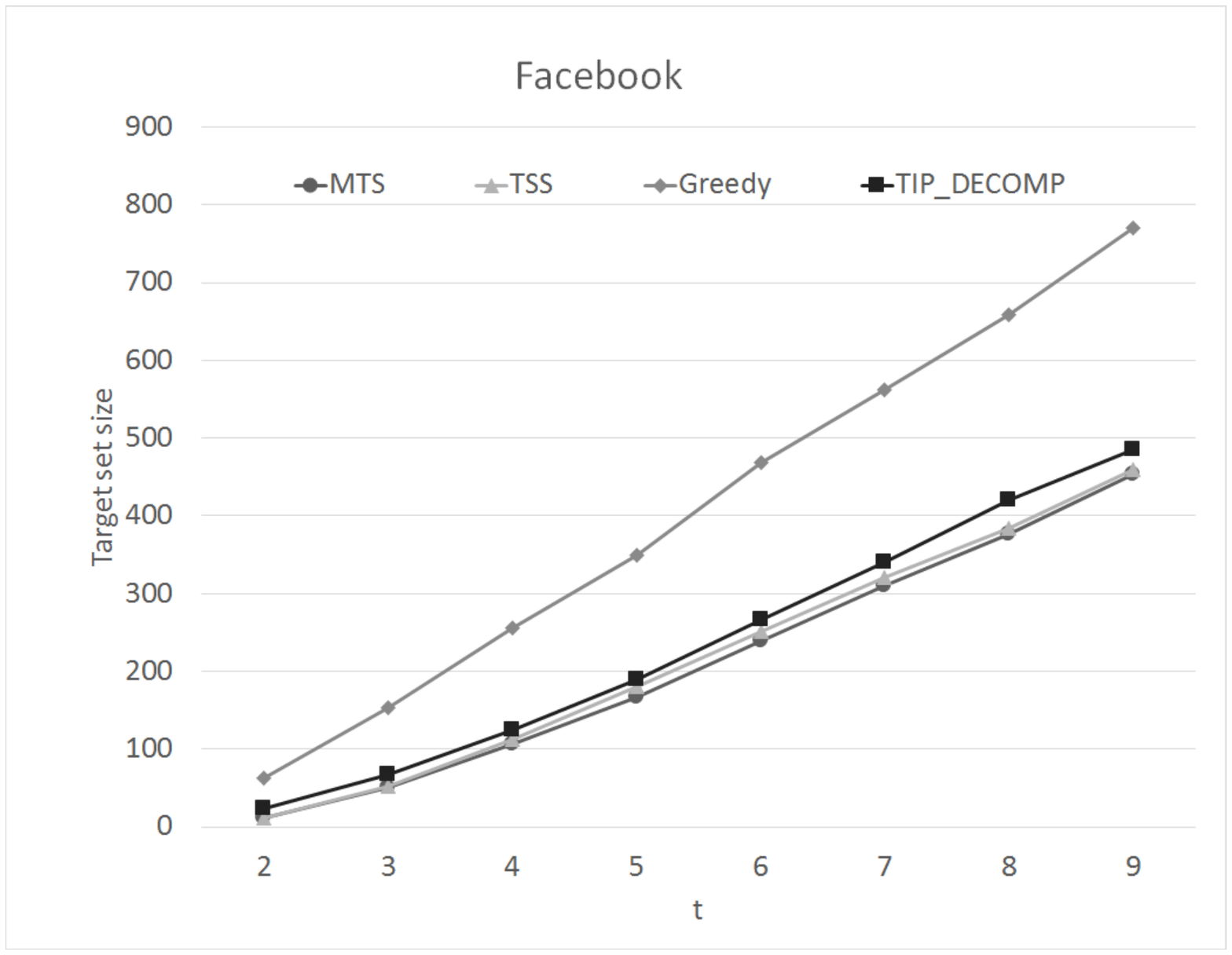}
		\includegraphics[width= 0.49\linewidth]{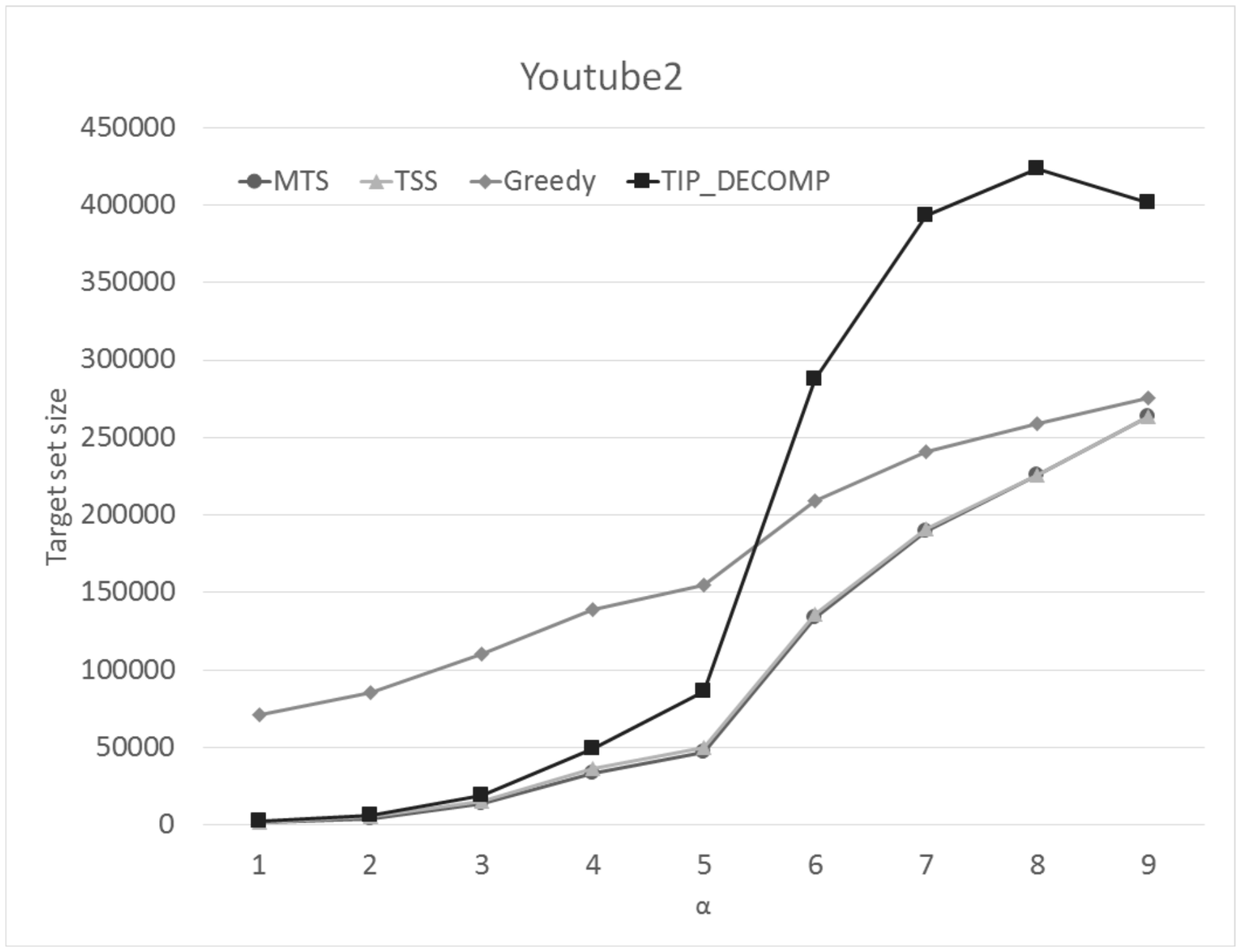}
		\includegraphics[width= 0.49\linewidth]{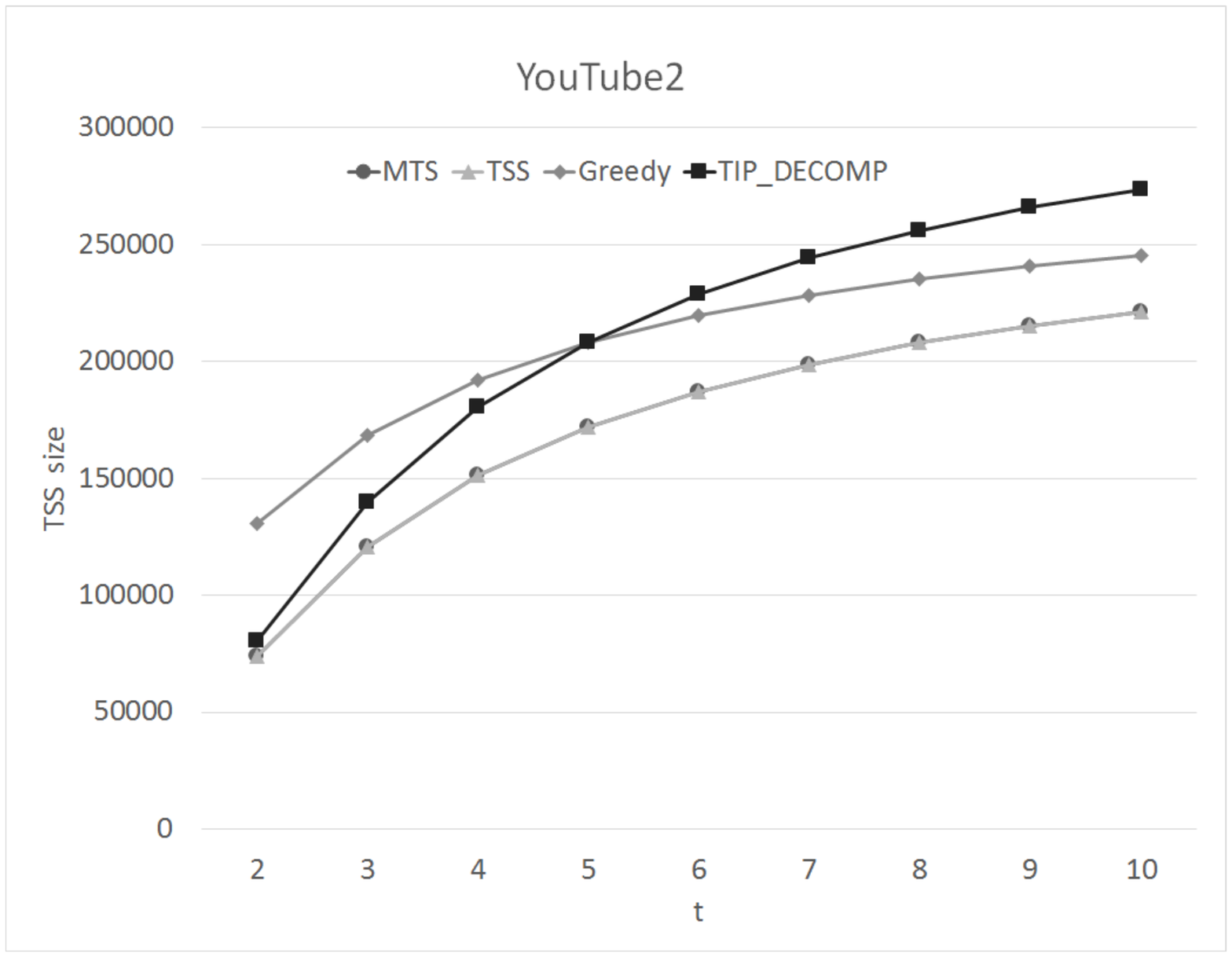}
				\caption{Constant and   Proportional Thresholds Results on undirected networks  (Last.FM, Facebook and Youtube2): For each network the results are reported in two separated charts: Proportional thresholds (left-side) and 
Constant thresholds (right-side).	\label{fig:Douban2}}
		\end{center}
\end{figure}

\begin{figure}[th!]
\begin{center}
		\includegraphics[width= 0.49\linewidth]{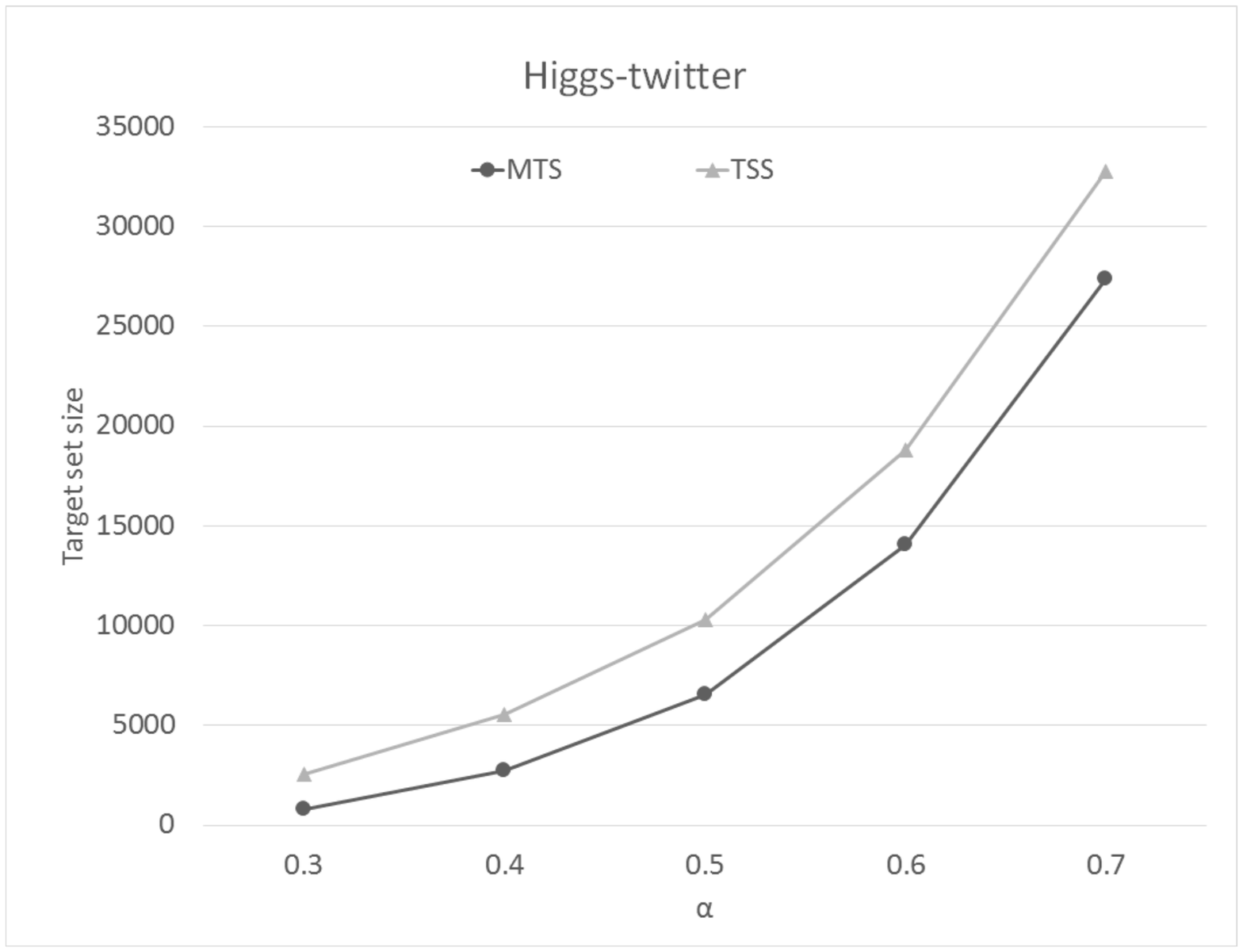}
		\includegraphics[width= 0.49\linewidth]{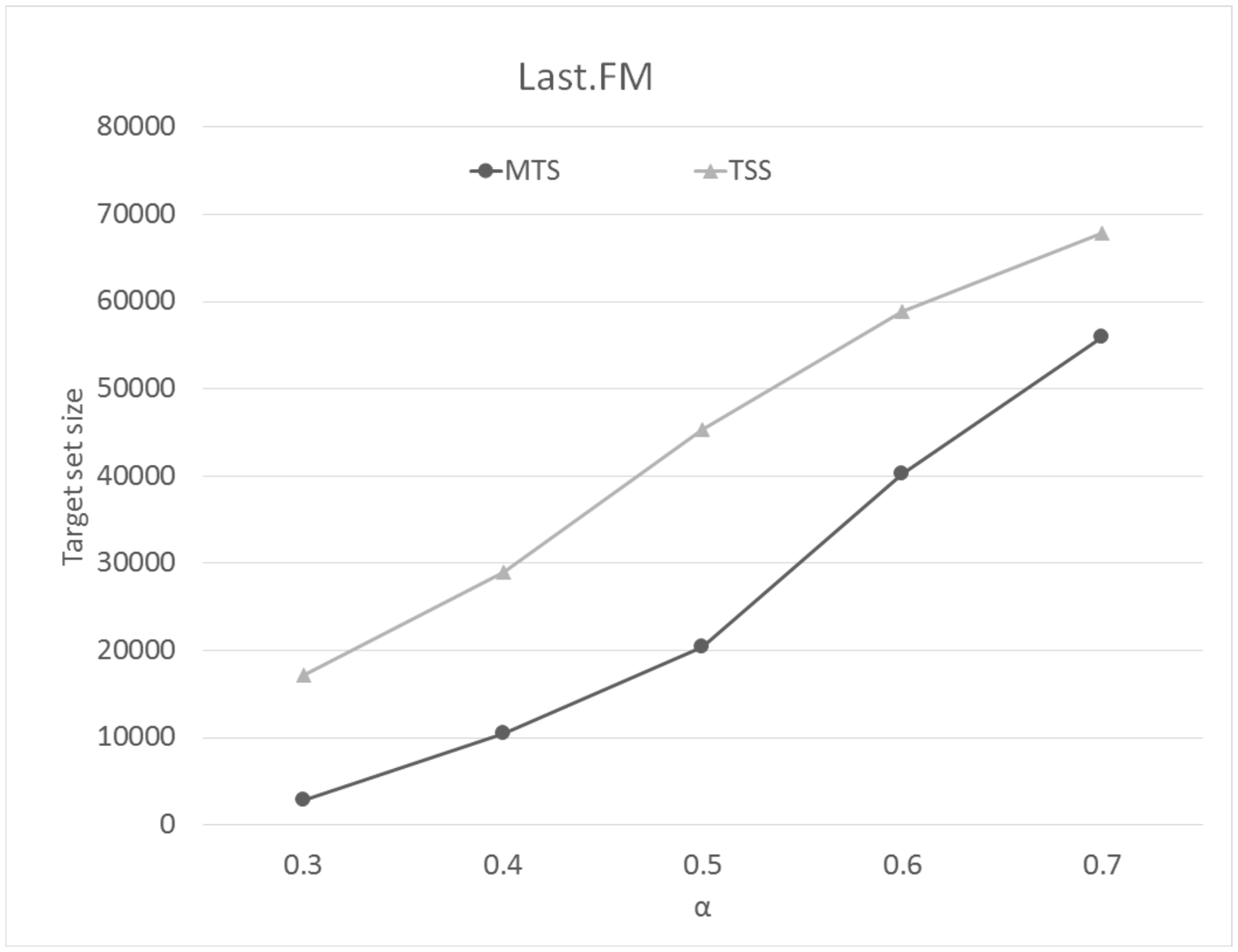}
		\caption{Comparison between MTS and TSS (proportional threshold with $\alpha \in  [0.3,0.7]$): (left) Higgs-twitter; (right) Last.FM.	\label{fig:gap}}
	\end{center}
\end{figure}

We present the results only for $10$ networks ($4$ Directed and $6$ Undirected); the
experiments performed on the other networks exhibit similar behaviors. 

 Analyzing the results from Figures  \ref{fig:amazon0302}-\ref{fig:Douban2}, we notice that
in all the considered cases our {MTS} algorithm always outperforms its competitors. 
{
The improvement is consistent with the results for random thresholds presented in the table \ref{randomtest}. The differences in terms of performance of the algorithms, in this case, depend on two factors: the structural characteristics of the network and  the thresholds. As noted previously, the largest differences are observed for intermediate values of the $\alpha$ parameter (for proportional thresholds) and for large values of the parameter $t$, when these do not exceed the average degree of the nodes (for constant thresholds).
}
In general, we provide the following observations: 
\begin{itemize}
	\item the TSS algorithm provides performance close to MTS but the gap increases in the proportional threshold case, especially for intermediate values of the parameter $\alpha$  (see Fig. \ref{fig:gap});
	\item  the Greedy algorithm performance improves with increasing thresholds;
	\item the TIP\_DECOMP performances worsen dramatically with increasing thresholds.
\end{itemize}

\subsection{Correlation between 
Network Modularity and Normalized Target Set Size}\label{sec:cor}
Analyzing the results from Figures \ref{fig:amazon0302} to  \ref{fig:Douban2}, we observe that the performance of the algorithms on different networks are influenced by the strength of communities of a network, measured by the modularity. Modularity is one measure of the structure of networks. Networks with high modularity have dense connections between the nodes within communities but sparse connections between nodes in different communities.
In order to better evaluate the correlation between the modularity and the performances of the algorithms (measured considering the normalized target set size, which corresponds to the size of the target set, provided by the algorithm, divided by the number of nodes in the network), we measured the correlation using a statistical metric: the Pearson product-moment Correlation Coefficient (PCC). PCC is one of the measures of correlation which quantifies the strength as well as direction
of the relationship between two variables. The correlation coefficient ranges from $-1$ (strong negative correlation) to $1$ (strong positive correlation). A value equal to  $0$ implies that there is no correlation between the variables. We computed the correlation PCC between network modularity  and normalized target set size, both with random and majority ($\alpha=0.5$) thresholds.
{ In particular, we considered two variables that are parametrized
by the class N of Networks (see Table \ref{net}), the algorithm $A \in\{$MTS, TSS, Greedy, TIP\_DECOMP$\}$, and the threshold function $F \in\{$Random, Majority$\}$. The
variable $M(N)$ denotes the network modularity; the variable
$T(N,A,F)$ denotes the normalized target set size. We observed
that, in all the considered cases, there is a moderate positive correlation
between modularity and normalized target set size (the PCC is between $
0.5$ and $0.7$). Figure \ref{fig:cor} presents the correlation values obtained.
The reasoning behind those results is that when the network is composed by strongly connected components (high modularity),  the influence  hardly  propagates from one community to another, thus    the size of the target set increases.
Figure \ref{fig:cor}  also shows that  the correlation does not depend on the threshold function, while it is more sensible on the results provided by the algorithms TSS and MTS. This results is probably due to the fact that 
the algorithms TSS and MTS are able to better exploit situations where the community structure of the networks allows a certain  influence between different communities.

We also performed similar analysis evaluating the correlation between the normalized target set size and the other network properties depicted in Table \ref{net}. Results show only a weak (the PCC is between $-0.5$ and $-0.3$) negative correlation  between the average degree and the normalized target set size. In all the other cases the PCC is between $-0.3$ and $+0.3$ (there is none or very weak correlation).}

\begin{figure}[ht!]
\begin{center}
		\includegraphics[ width= 0.75\linewidth]{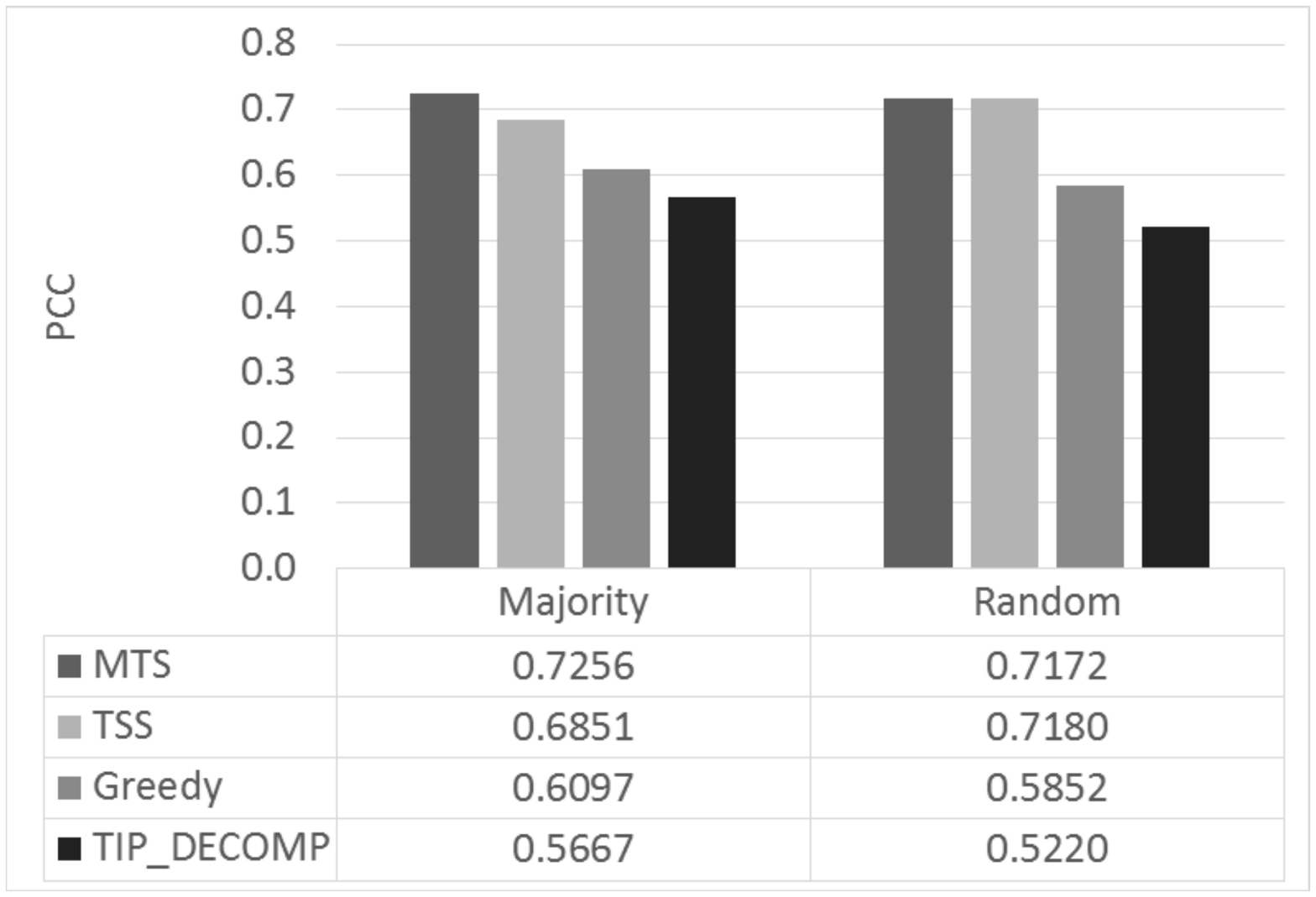}
		\caption{ Correlation between modularity and normalized target set size obtained using four algorithm with Random and Majority ($\alpha=0.5$) thresholds.	\label{fig:cor}}
\vspace*{-0.5truecm}
	\end{center}
\end{figure}

\section{Conclusion}

We considered the problem of selecting  a minimum size subset of nodes of  a network which can start an activation process that spreads to  
all the nodes of the network.  We presented a fast and simple
algorithm that  is optimal for several classes of graphs and   matches the general upper bound given in \cite{ABW-10,CGM+} on the cardinality of a minimum target set. Moreover,  on real life networks, it outperforms the other known heuristics for the same problem. 
Experimental results show that the performance of all the analyzed algorithms
correlates with the modularity of the analyzed network.
 This correlation is more sensible on the results provided by the
MTS algorithm. This results is probably due to the fact
that the proposed algorithms is able to better exploit situations when the
community structure of the networks allows a certain influence between
different communities.

There are many possible ways of extending our work.
We would be especially interested in discovering additional interesting  classes of graphs for which our algorithm is optimal or approximable within a small factor (with respect 
to the general $O(2^{\log^{1-\epsilon }|V|})$ inapproximability factor proved in \cite{Chen-09}).

\bibliography{TSS-doppiasoglia}{}
\bibliographystyle{plain}
\end{document}